%% file: main.tex
\documentclass{article} % For LaTeX2e
\usepackage{placeins}
\usepackage{conference,times}
\usepackage[colorlinks]{hyperref}
\definecolor{MyCustomBlue}{HTML}{3389DA}
\definecolor{MyCustomPink}{HTML}{EE4F89}
\definecolor{CustomBlue}{HTML}{2643CF}

\hypersetup{
    citecolor=MyCustomBlue,
    linkcolor=MyCustomPink,
    urlcolor=CustomBlue
}
\newcommand{\appendixentry}[2]{%
  \noindent\hyperref[#1]{\ref*{#1}\quad #2}%
  \dotfill
  \hyperref[#1]{\pageref*{#1}}%
  \par
}

\newcommand{\appendixsubentry}[2]{%
  \noindent\hspace{2.2em}\hyperref[#1]{\ref*{#1}\quad #2}%
  \dotfill
  \hyperref[#1]{\pageref*{#1}}%
  \par
}
\input{math_commands.tex}

\usepackage{wrapfig}
\newcommand{\coretodo}[1]{\textbf{[TO FILL: #1]}}
\usepackage{hyperref}
\usepackage{url}
\usepackage{tikz}
\usepackage{capt-of}
\usepackage{microtype}
\usepackage{amssymb}
\usepackage{amsthm}
\usepackage[table]{xcolor}
\definecolor{coreblue}{RGB}{220,240,250}

\definecolor{bestred}{RGB}{220,0,0}
\newtheorem{theorem}{Theorem}

\newtheorem{corollary}[theorem]{Corollary}
\newtheorem{remark}[theorem]{Remark}

\usepackage{tcolorbox}
\definecolor{corebestbg}{HTML}{FFE6E6}
\definecolor{coresecondbg}{HTML}{FFF0D9}
\definecolor{corevenuebg}{HTML}{D4ECF8}

\newcommand{\corebest}[1]{%
  \begingroup
  \setlength{\fboxsep}{1pt}%
  \colorbox{corebestbg}{\strut\textbf{#1}}%
  \endgroup
}
\newcommand{\coresecond}[1]{%
  \begingroup
  \setlength{\fboxsep}{1pt}%
  \colorbox{coresecondbg}{\strut\underline{#1}}%
  \endgroup
}

\newtcbox{\corevenuebox}{
  on line,
  boxrule=0pt,
  colback=corevenuebg,
  colframe=corevenuebg,
  arc=1.5pt,
  left=2pt,
  right=2pt,
  top=0.5pt,
  bottom=0.5pt,
  boxsep=0pt
}
\newcommand{\corevenuetag}[1]{%
  \nobreak\hspace{0.35em}%
  \corevenuebox{{\scriptsize\textcolor{black!80}{#1}}}%
}
\definecolor{theorempropbg}{RGB}{242,249,255}
\definecolor{theorempropborder}{RGB}{90,150,220}

\newtcolorbox{theorembox}{
  colback=theorempropbg,
  colframe=theorempropborder,
  boxrule=1.2pt,
  arc=2.5mm,
  left=0.5mm,
  right=0.5mm,
  top=0.3mm,
  bottom=0.3mm,
  before skip=5pt,
  after skip=5pt
}

\definecolor{propositionpropbg}{RGB}{244,252,247}
\definecolor{propositionpropborder}{RGB}{72,155,105}

\newtcolorbox{propositionbox}{
  colback=propositionpropbg,
  colframe=propositionpropborder,
  boxrule=1.2pt,
  arc=2.5mm,
  left=0.5mm,
  right=0.5mm,
  top=0.3mm,
  bottom=0.3mm,
  before skip=5pt,
  after skip=5pt
}

\title{One Distribution, Two Channels: Coverage-Calibrated KV Retention and Compensation}

\author{
Shuxin Liu\thanks{Equal contribution.}
\quad
Qing Liu\footnotemark[1]
\quad
Yi Du
\quad
Ou Wu\thanks{Corresponding author}
\\
Hangzhou Institute for Advanced Study, \\
University of Chinese Academy of Sciences
\\
{\normalfont\small
\texttt{\{liushuxin25,liuqing25\}@mails.ucas.ac.cn}
\quad
\texttt{wuou@ucas.ac.cn}}
}

\usepackage{booktabs}
\usepackage{graphicx}
\usepackage{algorithm}
\usepackage{algpseudocode}
\usepackage{float}
\usepackage{enumitem}
\usepackage{amsthm}
\newtheorem{proposition}{Proposition}
\usepackage{multirow}
\usepackage{comment}

\iclrfinalcopy
\begin{document}
\maketitle
\begin{abstract}
Long-context decoding is increasingly constrained by key--value (KV) cache memory and bandwidth. Existing fixed-budget compression methods typically separate retention from compensation, while a retention ranking specifies neither discarded attention mass nor the direction of induced output error. We start from an exact factorization: eviction error equals evicted attention \emph{mass} times the \emph{directional gap} between the evicted centroid and retained output, highlighting the importance of set-level coverage in retention and mass-preserving memory writing. We introduce \textbf{CORE} (\underline{\textbf{CO}}verage Calibration and Evicted-Mass \underline{\textbf{RE}}distribution for KV Cache), which distills an offline allocation combining query utility and log-determinant coverage into a lightweight cache-aware indexer. At inference, one calibrated distribution drives both channels: its Top-$B$ ordering retains complementary KV states, while its excluded allocation mass and conditional weights parameterize latent-memory writes without a separate write-weight predictor or online log-determinant evaluation. Our analysis provides a four-term pre-compensation error certificate, characterizes non-additive coverage interactions, and establishes mass-independent
write stability with hierarchical bounds through recurrent and
query-adaptive normalization. Across three backbones, CORE exceeds the strongest RULER baseline by up to 3.78 points at 90\% compression; LongBench and repeated-eviction evaluations further demonstrate strong effectiveness and decoding efficiency. \href{https://anonymous.4open.science/r/CORE-4D1F}{Code}.
\end{abstract}

\section{Introduction}
\label{sec:intro}

Long-context inference uses KV caching to reuse historical states during decoding~\citep{vaswani2017attention,shazeer2019fast,pope2023efficiently}, yet increasing context lengths make the cache a primary GPU-memory and bandwidth bottleneck~\citep{liu2024kivi,hooper2024kvquant,brandon2024reducing,behnam2025rocketkv}. Fixed-budget eviction mitigates this cost by retaining a subset of states~\citep{zhang2023h2o,li2024snapkv}, while recent work augments learned retention with latent memory for evicted information~\citep{yang2026indexmem}. Yet retention and compensation are typically treated separately, while a retention ranking alone does not fully characterize the error induced by eviction. This leaves two coupled questions: \textbf{which states should remain, and how should evicted information be represented?} Fig.~\ref{fig:overview} motivates a unified view of retention and compensation through the structure of eviction error.

Most eviction pipelines answer the first question with a ranking. A ranking determines an ordering, but it does not quantify \emph{how much} attention is removed or \emph{where} the induced output deviation points.  For an evicted set $E$ and retained cache $M$, this deviation for any query $q$ obeys \textbf{the exact identity}
\begin{equation}
\Delta o(q) := o_{\mathrm{full}}(q)-o_M(q)
=\pi_E(q)\bigl(\bar v_E(q)-o_M(q)\bigr),
\label{eq:intro-identity}
\end{equation}
where $\pi_E(q)$ denotes evicted attention mass,
$\bar v_E(q)$ its attention-weighted centroid, and
$o_M(q)$ the retained output. The identity follows from the normalized additive form of softmax attention; related work studies eviction-induced directional mismatch~\citep{li2026momentkvclosingdirectionalgap}. We use it as a design principle: retention must reduce the mass and directional factors jointly, whereas compensation must represent both. \textbf{A calibrated allocation distribution provides this shared interface:}
excluded mass measures discarded allocation, while its conditional restriction over evicted states specifies composition.

\begin{figure}[t]
	\centering
\includegraphics[width=0.95\linewidth]{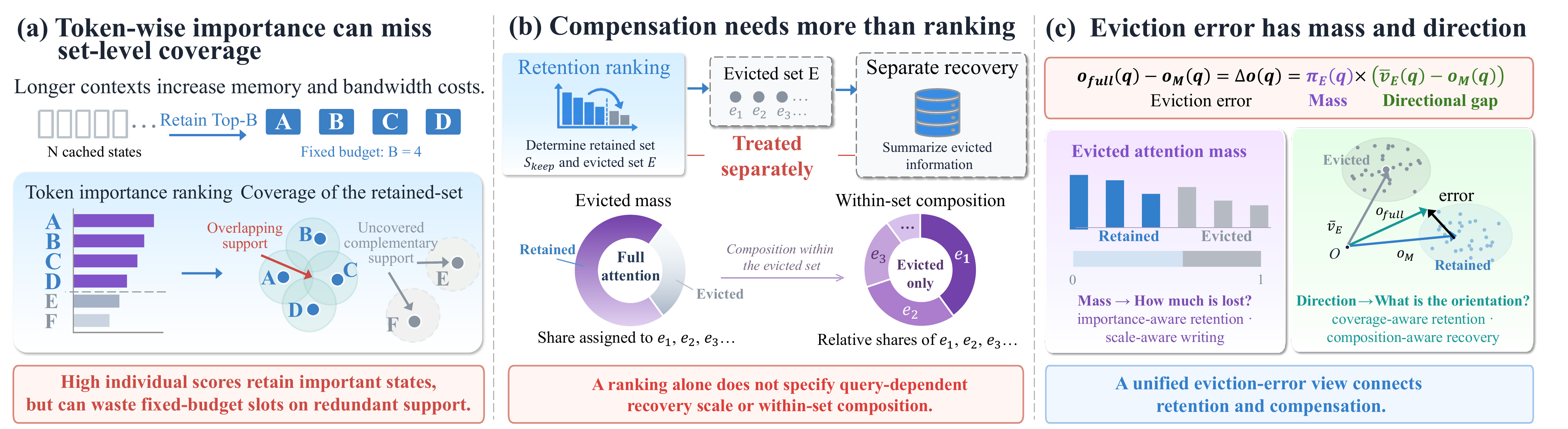}
\vspace{-0.15in}
	\caption{From separate retention and compensation to a unified eviction error perspective}
    \vspace{-0.2in}
    \label{fig:overview}
\end{figure}

The identity also exposes why token-wise selection without
accounting for coverage interactions is insufficient under a fixed budget. Importance scores emphasize mass and redundancy scores direction, but high individual scores do not ensure complementary coverage within the retained set. Redundant high-scoring states can therefore occupy multiple slots while adding little complementary information, making retention inherently a set-level coverage problem. Log-determinant objectives capture these interactions~\citep{kulesza2012determinantal}, but direct subset optimization repeatedly evaluates marginal gains~\citep{han2020map,hemmi2022lazy}; lightweight proxies approximate them cheaply but lose set dependence~\citep{park2026keydiff,cai2026r}, whereas cache-conditioned matrix criteria retain these interactions at online cost~\citep{yang2026rethinking}. Compensation faces a coupled challenge: a ranking carries neither scale nor composition required by Eq.~\ref{eq:intro-identity}, so separately parameterized memory can disagree with the retention policy that produced the eviction.

We introduce \textbf{CORE} (\underline{\textbf{CO}}verage Calibration and Evicted-Mass \underline{\textbf{RE}}distribution for KV Cache), which addresses the coupled challenges through one coverage-calibrated distribution. Offline, CORE combines query-conditioned utility with log-determinant coverage into a soft teacher allocation and distills it into a lightweight cache-aware indexer. At inference, this deployed distribution serves both channels: its Top-$B$ ordering determines retention, while its excluded mass and conditional allocation parameterize latent-memory writes without a separate compensation weighting rule or online matrix evaluation. Consistent with the scalar--direction identity, CORE reads from memory and gates the
result to compensate for eviction. Preserving excluded mass at the write level avoids sensitivity that
grows as this mass decreases. Our contributions are summarized as follows:
\begin{itemize}
[leftmargin=*,nosep]

\item
We start from an exact eviction-error identity that factorizes output
deviation into evicted mass and directional gap, exposing the joint
requirements of retention and compensation, showing normalized
aggregation readouts satisfy this exact factorization, and establishing
its converse characterization.

\item
We develop CORE, which distills offline set-level log-determinant coverage into a lightweight cache-aware indexer and uses one deployed distribution as two channels: Top-$B$ retention and evicted-mass redistribution, without an independent compensation weighting rule.

\item
Our theory establishes a four-term eviction certificate for
utility, coverage, allocation, and selection, characterizes non-additive
coverage interactions, and shows mass-preserving writes avoid inverse-mass
amplification, with stability via recurrent and query-adaptive normalization.

\item
Across long-context benchmarks and three backbones, CORE gains up to 3.78 points at 90\% compression while remaining robust to repeated decode-time compression with low online overhead.

\end{itemize}

\vspace{-0.05in}
\section{Related Work}
\vspace{-0.05in}
\label{sec:related}

\textbf{KV eviction.}
Fixed-budget KV management explores importance, recency, future-aware,
and context-conditioned retention signals
~\citep{zhang2023h2o,liu2023scissorhands,xiao2024efficient,
li2024snapkv,devoto2025expectedattention,ge2024model,
wu-etal-2025-scope,zhang-etal-2026-lazyeviction}.
Recent methods improve efficiency through adaptive policies,
retrieval, quantization, and learned retainers
~\citep{yang2024pyramidinfer,feng2026ada,zhou-etal-2025-dynamickv,
tang2024quest,ribar2024sparq,liu2024kivi,hooper2024kvquant,
huang2025locret,moschella2026learning,bui2026cache,
zhou2026lkv,yang2026indexmem}.
These methods provide diverse retention signals and budget policies. CORE focuses on distilling cache-conditioned coverage into an allocation determining evicted-state writing weights.

\textbf{Diversity and coverage in KV retention.}
Recent work augments token utility with diversity and redundancy reduction
to model retained-state interactions
~\citep{cai2026r,chang2026valueawarestochastickvcache,
kumari2024bumblebee,roy2026coveragedrivenkvcacheeviction,
tian2026every,liu2026mixing}.
Lightweight proxies trade fidelity for efficiency, while log-determinant
objectives preserve set-level coverage at the cost of marginal-gain
evaluation
~\citep{kulesza2012determinantal,han2020map,hemmi2022lazy}.
Information-theoretic formulations preserve cache capacity under fixed
budgets~\citep{yang2026rethinking}. Evaluation audits emphasize matching query visibility and physically enforced memory budgets~\citep{shi2026trustthemass}.
CORE transfers set-level coverage to deployment through
cache-conditioned distillation.

\textbf{Eviction compensation.}
Memory-backed methods reduce eviction loss via recallable pages,
clusters, and associative memories
~\citep{chen2024arkvale,liu2025clusterkv,
swain2026tensorcacheevictionconditionedassociative}; residual-cache
methods reconstruct omitted attention under fixed KV budgets
~\citep{zhan2026reskv}. CaM merges evicted values into retained
entries~\citep{zhang2024cam}.
Eq.~\ref{eq:intro-identity}
provides a common lens for comparison.  MOMENTKV uses evicted-set statistics to reduce directional mismatch,
whereas IndexMem separately parameterizes learned retention and
latent memory~\citep{li2026momentkvclosingdirectionalgap,yang2026indexmem}.
CORE couples retention and compensation through one allocation
distribution.

\vspace{-0.05in}
\section{Method}
\vspace{-0.05in}
\label{sec:method}

\begin{figure*}[t]
  \centering
  \includegraphics[width=0.94\textwidth]{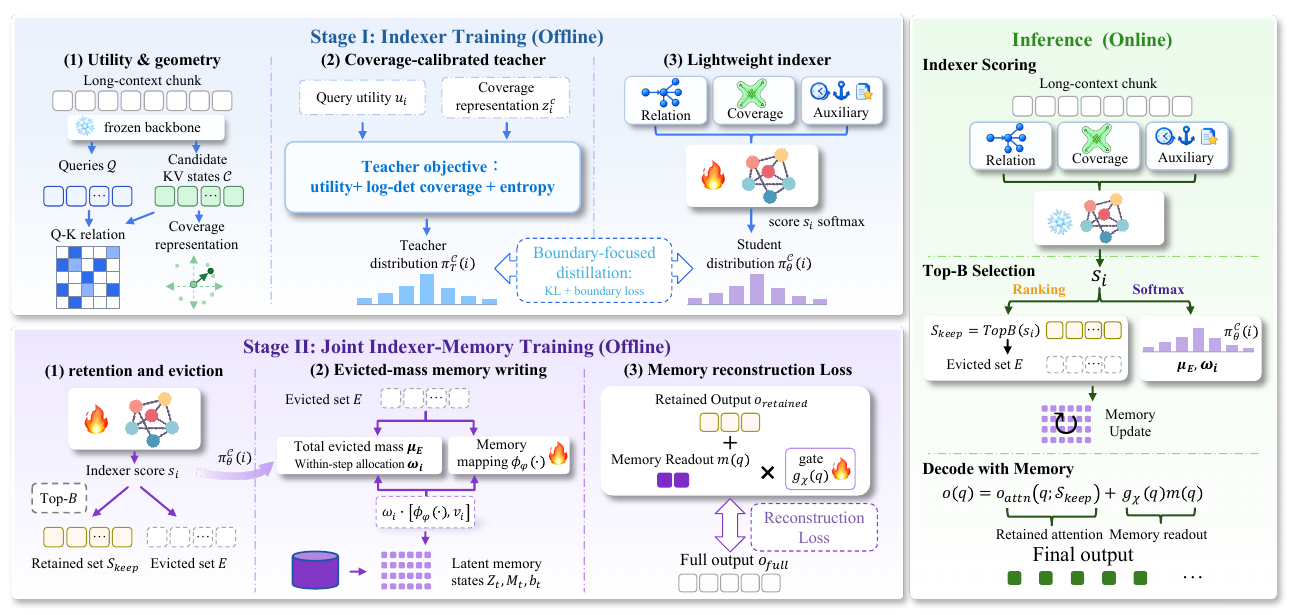}
  \vspace{-0.17in}
  \caption{
CORE couples KV retention and compensation via a coverage-calibrated distribution.
}
  \label{fig:core-overview}
  \vspace{-0.25in}
\end{figure*}

CORE follows three consequences of the eviction-error identity:
\emph{(i) mass-only retention reduces to importance-oriented scoring, while the directional term makes selection set-dependent; (ii) the excluded allocation determines write mass and composition; and (iii) scalar--direction residuals motivate gated correction.}
These consequences define the principles governing retention and compensation, so CORE distills offline log-determinant coverage into one lightweight allocation whose Top-$B$ and excluded restrictions drive retention and memory writing.
See Fig.~\ref{fig:core-overview} for an overview.

\vspace{-0.05in}
\subsection{Coverage-Calibrated Soft Allocation and Retention Distillation}
\label{sec:retention}
\vspace{-0.05in}
At Transformer layer $\ell$, let $\mathcal C=\{1,\ldots,N\}$ denote the
candidate KV positions, with $(k_i^{(m)},v_i^{(m)})$ the key--value states
for position $i$ and head $m$. Given a query set
$\mathcal Q$ and retention budget $B<N$, CORE learns a
coverage-calibrated allocation over $\mathcal C$; sorting this allocation
supplies the retained set.

\textbf{Query utility and key-space representation.}
We first derive an attention-based utility signal from the frozen backbone
to quantify the query-dependent importance of candidate KV state $i$:
\vspace{-0.05in}
\begin{equation}
a_{t,h,i}
=
\frac{(q_t^{(h)})^\top k_i^{(g(h))}}
{\sqrt{d_{\mathrm{head}}}},
\qquad
u_i = \max\nolimits_{t\in\mathcal Q:\,i\le t,\;m}
\operatorname{TopMean}_{h:g(h)=m} a_{t,h,i},
\label{eq:query-utility}
\vspace{-0.1in}
\end{equation}
where $q_t^{(h)}$ denotes the query state at position $t$ with head $h$,
and $g(h)$ maps query heads to shared KV heads.
$\operatorname{TopMean}$ aggregates the largest $K_{\mathrm{head}}$
responses within each KV group to preserve strong query--key matches
while reducing head-specific noise. During prefill, $\mathcal Q$ contains
prompt queries; during decode-time compression, it contains queries
generated in the current compression interval.

Query utility measures relevance to queries but does not
encode interactions among candidates. We therefore construct a compact
key-space representation in which set-level coverage is estimated:
\vspace{-0.05in}
\begin{equation}
\bar{k}_i
=
\operatorname{Concat}_{m=1}^{H_{\mathrm{kv}}}
\left[
\frac{k_i^{(m)}}{\|k_i^{(m)}\|_2+\epsilon}
\right],
\qquad
z_i
=
\operatorname{Normalize}
\left(P_k^{(\ell)}\bar{k}_i\right)
\in\mathbb{R}^{d_z},
\label{eq:coverage-representation}
\end{equation}
where $P_k^{(\ell)}$ is a fixed, layer-specific random projection into a
compact key space, and
$z_i^c=z_i-\frac{1}{N}\sum_j z_j$ centers each representation relative to
the current candidate set.

\textbf{Coverage-calibrated teacher.}
Using the query utility $u_i$ and cache-centered representation $z_i^c$,
we define a soft allocation that jointly captures query relevance and
set-level directional coverage:
\begin{equation}
\begin{aligned}
w^\star
=
\arg\max_{w\in\Delta_N}
\quad
&
\sum\nolimits_{i=1}^{N} w_i u_i
+
\beta
\log\det
\left(
I+\alpha\sum \nolimits_i w_i z_i^c(z_i^c)^\top
\right)
-
\tau_{\mathrm{ent}}
\sum\nolimits_{i=1}^{N} w_i\log w_i ,
\end{aligned}
\label{eq:teacher-objective}
\end{equation}
where $\Delta_N=\{w\in\mathbb{R}_+^N:\sum_i w_i=1\}$.
The objective unifies relevance, cache-conditioned coverage, and
entropy in one allocation. CORE approximates its optimizer
with $R$ fixed-point updates:
{\small
\begin{equation}
\begin{aligned}
\Sigma^{(r)}
&=
I+\alpha\sum\nolimits_i w_i^{(r)}z_i^c(z_i^c)^\top,
&
d_i^{(r)}
&=
\alpha(z_i^c)^\top(\Sigma^{(r)})^{-1}z_i^c,
&
w_i^{(r+1)}
&=
\operatorname{softmax}_i
\left(
\frac{u_i+\beta d_i^{(r)}}{\tau_{\mathrm{ent}}}
\right).
\end{aligned}
\label{eq:teacher-update}
\end{equation}}
Starting from a uniform allocation, CORE sets
$\pi_T^{\mathcal C}(i)=w_i^{(R)}$. Since $\Sigma^{(r)}$ captures
directions weighted by allocation,
$d_i^{(r)}$ provides a cache-conditioned coverage marginal that discounts
redundancy and favors under-represented directions. Derivation and
finite-step analysis see Appendix~\ref{app:teacher-solver}.

\textbf{Feature-based indexer.}
Evaluating the coverage teacher online would require repeated
cache-dependent matrix operations. CORE distills its set-level
structure into a lightweight cache-aware indexer. For each candidate $i$,
we construct
$\xi_i=[\xi_i^{\mathrm{rel}};\xi_i^{\mathrm{cov}};\xi_i^{\mathrm{aux}}]$,
where the three groups encode query relevance, cache geometry, and
deployment state. A lightweight scorer $f_\theta$ produces
\begin{equation}
s_i=f_\theta(\xi_i),
\qquad
\pi_\theta^{\mathcal C}(i)
=
\frac{\exp(s_i/\tau_M)}
{\sum_{j\in\mathcal C}\exp(s_j/\tau_M)},
\label{eq:student-distribution}
\end{equation}
For KL distillation, protected sinks are masked and teacher and student
allocations are renormalized over the remaining candidates; notation is
unchanged in the loss below.

\emph{Relation features.}
Cache-baseline normalization of query--key responses yields a relevance signal:
\begin{equation}
r_{t,h,i}
=
\log
\left(
|\mathcal C_t|
\frac{\exp(a_{t,h,i})}
{\sum_{j\in\mathcal C_t}\exp(a_{t,h,j})}
+\epsilon
\right),
\qquad
R_{t,m,i}
=
\operatorname{TopMean}_{h:g(h)=m}r_{t,h,i}.
\label{eq:relation-signal}
\end{equation}
Statistics of $R_{t,m,i}$ across queries and KV groups form
$\xi_i^{\mathrm{rel}}$ on that cache-relative scale.

\emph{Coverage features.}
To capture cache geometry without reconstructing the teacher matrix online,
CORE uses a fixed spherical prototype bank
$\mathcal P=\{p_1,\ldots,p_{M_p}\}$ over cache-centered directions.
With $\widehat z_i^c=z_i^c/\max(\|z_i^c\|_2,\epsilon_{\mathrm{geo}})$,
for reliable directions,
\begin{equation}
\xi_i^{\mathrm{cov}}
=
\left[
\log(1+n_{i,1}),
\;
1-(\widehat z_i^c)^\top p_{i,1},
\;
(\widehat z_i^c)^\top p_{i,1}
-
(\widehat z_i^c)^\top p_{i,2},
\;
g_i^{\mathrm{dev}}
\right],
\label{eq:coverage-features}
\end{equation}
where $p_{i,1}$ and $p_{i,2}$ are the two nearest prototypes and
$n_{i,1}$ is the occupancy of the nearest prototype. Together these
features summarize local crowding, directional residual, assignment
margin, and global cache-relative deviation, providing the compact
interface through which cache-conditioned coverage is distilled. Finally, $\xi_i^{\mathrm{aux}}$ records token age, sink or special-token
status, value energy, and prefill/decode phase. Full feature definitions
and indexer details appear in
Appendix~\ref{app:feature-details}.

\textbf{Boundary-focused distillation.}
Under this convention, CORE matches teacher and student allocations by KL
and adds boundary supervision at the discrete Top-$B$ cutoff:
{\small
\begin{equation}
\begin{aligned}
\mathcal{L}_{\mathrm{idx}}
&=
D_{\mathrm{KL}}
\left(
\pi_T^{\mathcal C}
\,\middle\|\,
\pi_\theta^{\mathcal C}
\right)
+
\lambda_{\mathrm{bd}}\mathcal{L}_{\mathrm{bd}},
&
\mathcal{L}_{\mathrm{bd}}=
\frac{1}{|P_{\mathrm{bd}}||N_{\mathrm{bd}}|}
\sum\nolimits_{i\in P_{\mathrm{bd}}}
\sum\nolimits_{j\in N_{\mathrm{bd}}}
\left[
m_{\mathrm{bd}}
-
(\widetilde{s}_i-\widetilde{s}_j)
\right]_+ .
\end{aligned}
\label{eq:indexer-loss}
\end{equation}}
Here $\widetilde{s}_i=s_i/\tau_M$, and
$P_{\mathrm{bd}},N_{\mathrm{bd}}$ are teacher-ranked windows
bracketing the Top-$B$ cutoff.
The boundary term penalizes local retain--evict inversions
and complements KL matching.
Prefill averages boundary losses over compression ratios;
single-budget decode supervision uses the ratio sampled
for the current optimizer step.
Details and retention-fidelity analysis appear in
Appendices~\ref{app:boundary-loss},~\ref{app:joint-fidelity}.

\subsection{Evicted-Mass Redistribution}
\label{sec:emr}

The compensation channel follows the mass--direction motivation of
Eq.~\ref{eq:intro-identity}. CORE couples retention and compensation
through one allocation: Top-$B$ selects the retained set
$\mathcal S_{\mathrm{keep}}=\operatorname{TopB}
(\{s_i\}_{i\in\mathcal C})$,
while the allocation restricted to the evicted set
$E=\mathcal C\setminus\mathcal S_{\mathrm{keep}}$
determines memory write. Accordingly,
\begin{equation}
\begin{aligned}
\mu_E
&=
\sum \nolimits_{i\in E}\pi_\theta^{\mathcal C}(i),
&
\omega_i
=
\pi_\theta^{\mathcal C}(i)/{\mu_E},
\quad i\in E .
\end{aligned}
\label{eq:evicted-mass}
\end{equation}
$\mu_E$ and $\omega_i$ specify excluded allocation mass and
conditional composition under the query-informed student distribution. Since
$\mu_E\omega_i=\pi_\theta^{\mathcal C}(i)$, the memory write preserves
excluded allocation:
$\mu_E\bar M_E=\sum_{i\in E}\pi_\theta^{\mathcal C}(i)
\phi_\varphi(k_i)v_i^\top$.
Appendix~\ref{app:proof-allocation-decomposition}
shows the mass--composition factorization is unique.

A lightweight map $\phi_\varphi$ provides memory addressing. At each
compression step $t$, the evicted summary is accumulated with its excluded
mass:
{\small
\begin{equation}
\begin{aligned}
\bar M_E
&=
\sum \nolimits_{i\in E}
\omega_i\phi_\varphi(k_i)v_i^\top,
\quad
\bar b_E
=
\sum \nolimits_{i\in E}
\omega_i
\left[
\phi_\varphi(k_i)\odot\phi_\varphi(k_i)
\right],
\\
Z_t
&=
\gamma_MZ_{t-1}+\mu_{E_t},
\quad
M_t=
\frac{
\gamma_MZ_{t-1}M_{t-1}
+\mu_{E_t}\bar M_{E_t}
}{
Z_t+\epsilon
},
\quad
b_t
=
\frac{
\gamma_MZ_{t-1}b_{t-1}
+\mu_{E_t}\bar b_{E_t}
}{
Z_t+\epsilon
}.
\end{aligned}
\label{eq:memory-update}
\end{equation}}
Here $\mu_{E_t}$ sets relative
event weight, $\gamma_M$ discounts prior events,
and $Z_t$ accumulates mass,
while $(M_t,b_t)$ retain normalized evicted composition.

For a subsequent query $q$, CORE forms a normalized query-dependent
memory readout:
\begin{equation}
m(q)
=
\frac{
\phi_\varphi(q)^\top M_t
}{
\sqrt{
\left(\phi_\varphi(q)\odot\phi_\varphi(q)\right)^\top b_t+\epsilon
}
},
\qquad
o(q)
=
o_{\mathrm{attn}}
\left(q;\mathcal S_{\mathrm{keep}}\right)
+
g_\chi(q)m(q),
\label{eq:memory-readout}
\end{equation}
where $g_\chi(q)\in[0,1]$ is a learned scalar gate. Thus the shared
allocation determines within-event composition and recurrent event weight,
while $g_\chi(q)m(q)$ reconstructs the residual.

The memory modules are trained on the residual between
full and retained attention outputs:
\begin{equation}
\mathcal{L}_{\mathrm{mem}}
=
\mathbb{E}_{q}
\left[
\left\|
o_{\mathrm{full}}(q)
-
o_{\mathrm{attn}}
\left(q;\mathcal S_{\mathrm{keep}}\right)
-
g_\chi(q)m(q)
\right\|_2^2
\right].
\label{eq:memory-loss}
\vspace{-0.02in}
\end{equation}

With offline full-attention supervision, $\phi_\varphi$ and $g_\chi$
learn to reconstruct Top-$B$ retention residuals.

\subsection{Training, Inference, and Complexity}
\label{sec:procedures}

\begin{wrapfigure}{r}{0.55\columnwidth}
\vspace{-0.35in}

\begin{minipage}{\linewidth}
\begin{algorithm}[H]
\caption{Offline CORE Training}
\label{alg:offline-training}

\begingroup
\footnotesize
\linespread{0.95}\selectfont
\algrenewcommand\algorithmicindent{0.8em}
\begin{algorithmic}[1]

\Require KV states $\mathcal C$, queries $\mathcal Q$,
compression ratios $\mathcal R$

\Statex \textbf{Stage I: Allocation distillation}
\State Compute $\pi_T^{\mathcal C}=w^{(R)}$ using
Eq.~\ref{eq:teacher-update}
\State Compute $s_i=f_\theta(\xi_i)$ and $\pi_\theta^{\mathcal C}$
using Eq.~\ref{eq:student-distribution};
minimize $\mathcal L_{\mathrm{idx}}$

\Statex \textbf{Stage II: Joint retention and compensation training}
\State Sample $r\sim\operatorname{Uniform}(\mathcal R)$
per optimizer step
\For{each prefix of length $N$}
\State Set $B=\lfloor(1-r)N\rfloor$; retain Top-$B$
\State Compute $\mathcal L_{\mathrm{mem}}$ from evicted states
\EndFor
\State Optimize $\theta$ with $\mathcal L_{\mathrm{idx}}$
and $(\phi_\varphi,g_\chi)$ with $\mathcal L_{\mathrm{mem}}$

\end{algorithmic}

\endgroup

\end{algorithm}
\end{minipage}

\vspace{-0.15in}
\end{wrapfigure}

Stage I distills the coverage-calibrated teacher
$\pi_T^{\mathcal C}$ into an indexer predicting
$\pi_\theta^{\mathcal C}$.
This allocation couples both channels: its ranking determines
Top-$B$ retention, while its excluded mass and conditional
weights determine memory writes.
Stage II continues indexer distillation with $\mathcal L_{\mathrm{idx}}$
and trains memory maps and gates with $\mathcal L_{\mathrm{mem}}$
to reconstruct attention residuals induced by that retention policy.
Algorithm~\ref{alg:offline-training} summarizes this two-stage procedure.

\begin{wrapfigure}{r}{1.27\linewidth}
\vspace{-0.15in}
\centering

\begin{minipage}{\linewidth}

\begin{algorithm}[H]
\caption{Online Retention and Compensation}
\label{alg:online-core}

\begingroup
\footnotesize
\linespread{0.95}\selectfont
\algrenewcommand\algorithmicindent{0.8em}
\begin{algorithmic}[1]

\Require Candidate KV states $\mathcal C$, observed queries
$\mathcal Q_{\mathrm{obs}}$, subsequent query $q$,
budget $B$, memory states $(M,b,Z)$

\State Predict allocation
$\pi_\theta^{\mathcal C}(i)$ for $i\in\mathcal C$
using $\mathcal Q_{\mathrm{obs}}$

\State Select
$\mathcal S_{\mathrm{keep}}
\gets
\operatorname{TopB}
(\{\pi_\theta^{\mathcal C}(i)\})$
under the protected-sink convention,
$E\gets\mathcal C\setminus\mathcal S_{\mathrm{keep}}$

\State Compute
$\mu_E,\{\omega_i\}_{i\in E}$
from $\pi_\theta^{\mathcal C}$
and update $(M,b,Z)$

\State Compute compensated attention for $q$
using Eq.~\ref{eq:memory-readout}

\end{algorithmic}
\endgroup

\end{algorithm}

\end{minipage}

\vspace{-0.2in}
\end{wrapfigure}

\textbf{Deployment consistency.}
The teacher uses chunk-level query context only for constructing offline
targets. At deployment, $\mathcal Q_{\mathrm{obs}}$ contains prompt queries
during prefill and only causally available queries generated within
the current decoding interval thereafter. In
Stage II, each held-out pseudo-decode query is supervised by memory written
from the preceding prefix, so current and future KV states are excluded.

At inference in Algorithm~\ref{alg:online-core}: CORE predicts
the allocation, performs Top-$B$ retention, updates the evicted-mass
memory, and applies compensated attention through
Eq.~\ref{eq:memory-readout}. One teacher iteration costs
$\mathcal O(Nd_z^2+d_z^3)$, with total cost
$\mathcal O(R(Nd_z^2+d_z^3))$ over $R$ iterations. Streaming Q--K
statistics require
$\mathcal O(|\mathcal Q|Nd_{\mathrm{head}})$, while scoring and
memory updates remain linear in $N$ before Top-$B$ selection.

\section{Theoretical Analysis}
\label{sec:theory}

% \subsection{Scope of the identity.}
% \label{sec:scope}

\textbf{Scope of the identity.}
Eq.~\ref{eq:intro-identity} is not specific to softmax attention. It
holds for \emph{normalized aggregation readouts} $o(S)=N(S)/Z(S)$, where
$N(S)$ and $Z(S)>0$ are additive over disjoint subsets. This class
includes softmax and sparse attention, retrieval caches, and gated
aggregation.
\begin{theorembox}
\begin{theorem}[Mass--direction characterization]
\label{thm:characterization}
Every normalized aggregation readout admits the exact decomposition
$o(M\cup E)-o(M)=\pi_E\bigl(o(E)-o(M)\bigr)$
for disjoint nonempty $M,E$; under condition (A3) in
Appendix~\ref{app:scope-characterization}, the converse holds.
The representation yields
$\pi_E=\frac{Z(E)}{Z(M)+Z(E)}$ as the relative mass of the evicted subset
and $o(E)$ as its normalized readout.
\end{theorem}
\end{theorembox}
\vspace{-0.05in}
The mass--direction decomposition is universal for normalized aggregation and,
under (A3), characterizes this readout class.
Proofs and boundary cases are in Appendix~\ref{app:scope}.

\textbf{One Certificate, Two Channels.}
Fix a head, cache, and eviction-time allocations, with deployed set
$E\subsetneq\mathcal C$.
For any specified query distribution $\mathcal Q$, assume
$\|v_i\|_2\le B_v$, $\beta\ge0$, $\tau_{\mathrm{ent}}>0$, and write
$\mu_E=\pi_\theta^{\mathcal C}(E)$,
$\bar p=\mathbb E_{q\sim\mathcal Q}[p_{\mathrm{attn}}(q)]$ and
$\pi_u=\sigma(u/\tau_{\mathrm{ent}})$.
Here $p_{\mathrm{attn}}(q)$ denotes full-cache attention.
Let $\delta_{\mathrm{cov}}=\|\pi_T^{\mathcal C}-\pi_u\|_1$
and $d=d(w^{(R-1)})$.
\vspace{-0.03in}
\begin{theorembox}
\begin{theorem}[Expected eviction certificate]
\label{thm:certificate}
The expected pre-compensation eviction error satisfies
\vspace{-0.08in}
\begin{equation}
\mathbb E_{q\sim\mathcal Q}\|\Delta o(q)\|_2
\le
2B_v\min\!\left\{
1,\,\mu_E+\frac12\left[
\|\pi_\theta^{\mathcal C}-\pi_T^{\mathcal C}\|_1
+\delta_{\mathrm{cov}}+\|\pi_u-\bar p\|_1
\right]\right\}.
\label{eq:eviction-certificate}
\vspace{-0.08in}
\end{equation}
Moreover, $\delta_{\mathrm{cov}}\le
\min\{2,\beta\operatorname{osc}(d)/(2\tau_{\mathrm{ent}})\}$,
where $\operatorname{osc}(d)=\max_i d_i-\min_i d_i$.
\end{theorem}
\end{theorembox}
\vspace{-0.03in}
This bound uses $2B_v$ for direction and separates excluded
mass, allocation mismatch, coverage-induced shift, and utility mismatch.
Unknown $\bar p$ prevents direct online evaluation; directional benefits
are assessed empirically (Sec.~\ref{sec:mechanistic-analysis}).
Under shared constraints, Top-$B$ minimizes excluded allocation mass;
$\pi_i=\mu_E\omega_i$ uniquely factorizes excluded weights for $\mu_E>0$.
Teacher--student allocation error bounds retention regret
and excluded-mass discrepancy
(Appendix~\ref{app:joint-fidelity}).
Proofs and extensions appear in Appendices~\ref{app:proof-certificate}, \ref{app:decode-certificate}, and
\ref{app:proof-cache-conditioned-coverage}.

\textbf{Mass-Preserving Memory Stability.}
Shared allocation preserves excluded mass alongside conditional
evicted-token distributions.
Let $p,q\in\Delta_N$ induce evicted sets $E_p,E_q$. Define
$\eta(p,q)=\|p-q\|_1+
\min\{p(E_p\triangle E_q),q(E_p\triangle E_q)\}$
and $R(r,E)=\sum_{i\in E}r_iA_i$ for memory atoms
$\|A_i\|\leq C$.
\vspace{-0.1in}
\begin{propositionbox}
\begin{proposition}[Evicted-mass write stability]
\label{prop:memory-stability}
For $\mu_{\min}=\min\{p(E_p),q(E_q)\}>0$, mass-preserving
writes satisfy
$\|R(p,E_p)-R(q,E_q)\|\leq C\,\eta(p,q)$;
conditional-only writes satisfy
\vspace{-0.08in}
\begin{equation}
\left\|
\frac{R(p,E_p)}{p(E_p)}
-
\frac{R(q,E_q)}{q(E_q)}
\right\|
\leq
\frac{2C\,\eta(p,q)}{\mu_{\min}}.
\label{eq:conditional-write-stability}
\end{equation}
\end{proposition}
\end{propositionbox}
\vspace{-0.03in}
Proposition~\ref{prop:memory-stability} isolates the write-level mechanism:
preserving allocation scale removes the inverse-mass amplification caused
by conditional normalization.
Appendix~\ref{app:recurrent-memory-characterization} propagates this
stability through recurrent and query-adaptive normalization, explicitly
separating the additional accumulated-mass and readout-denominator factors.

\begin{table*}[t]
\centering
\caption{RULER results across KV compression ratios.
0\% denotes full KV cache.
Within each model and compressed setting, red/bold and
orange/underlined indicate best and second-best results.}
\label{tab:main-results}

\scriptsize
\setlength{\tabcolsep}{3pt}
\setlength{\aboverulesep}{1.5pt}
\setlength{\belowrulesep}{1.5pt}
\renewcommand{\arraystretch}{1.0}

% Individual number highlights with matched box heights and no row expansion.
\renewcommand{\corebest}[1]{%
  \begingroup
  \setlength{\fboxsep}{0.8pt}%
  \smash{\colorbox{corebestbg}{%
    \vphantom{\underline{\textbf{0.00}}}\textbf{#1}%
  }}%
  \endgroup
}
\renewcommand{\coresecond}[1]{%
  \begingroup
  \setlength{\fboxsep}{0.8pt}%
  \smash{\colorbox{coresecondbg}{%
    \vphantom{\underline{\textbf{0.00}}}\underline{#1}%
  }}%
  \endgroup
}

% Preserve the venue badge, but let the table strut determine row height.
\renewcommand{\corevenuetag}[1]{%
  \nobreak\hspace{0.35em}%
  \smash{\corevenuebox{{\scriptsize\textcolor{black!80}{#1}}}}%
}

% Account for the cmidrule and its spacing between the two header rows.
\newcommand{\coretablehead}[1]{%
  \multirow{2}{*}%
    [-\dimexpr(\aboverulesep+\belowrulesep+\cmidrulewidth)/2\relax]%
    {\textbf{#1}}%
}

\resizebox{\textwidth}{!}{%
\begin{tabular}{@{}c|c|ccccc|ccccc@{}}
\toprule
\coretablehead{Model}
& \coretablehead{Method}
& \multicolumn{5}{c|}{\textbf{RULER-4K Avg.}}
& \multicolumn{5}{c@{}}{\textbf{RULER-16K Avg.}} \\
\cmidrule(lr){3-7}
\cmidrule(l){8-12}
& & \textbf{0\%} & \textbf{25\%} & \textbf{50\%}
  & \textbf{75\%} & \textbf{90\%}
  & \textbf{0\%} & \textbf{25\%} & \textbf{50\%}
  & \textbf{75\%} & \textbf{90\%} \\
\midrule

% ==================== Mistral-7B ====================
\multirow{9}{*}{\rotatebox[origin=c]{90}{\textbf{Mistral-7B}}}
& H$_2$O\corevenuetag{NeurIPS'23}
& 93.29 & 68.65 & 49.13 & 34.36 & 21.15
& 80.95 & 52.49 & 37.11 & 24.55 & 15.74 \\

& SnapKV\corevenuetag{NeurIPS'24}
& 93.29 & 58.64 & 42.71 & 33.16 & 20.91
& 80.95 & 44.61 & 32.06 & 24.57 & 19.04 \\

& PyramidKV\corevenuetag{COLM'25}
& 93.29 & 62.14 & 41.68 & 30.92 & 20.92
& 80.95 & 49.04 & 34.41 & 23.50 & 17.82 \\

& Expected Attention\corevenuetag{arXiv'25}
& 93.29 & 72.60 & 48.57 & 28.59 & 13.74
& 80.95 & 59.06 & 43.28 & 28.94 & 15.79 \\

& KeyDiff\corevenuetag{NeurIPS'25}
& 93.29 & \coresecond{87.09} & 81.58
& \coresecond{70.49} & \coresecond{52.17}
& 80.95 & 75.26 & \coresecond{70.79} & 65.80 & 50.40 \\

& CAPKV\corevenuetag{ICML'26}
& 93.29 & 86.35 & 70.55 & 55.80 & 36.27
& 80.95 & \corebest{75.46} & 57.63 & 46.19 & 36.64 \\

& IndexMem\corevenuetag{ICML'26}
& 93.29 & 84.21 & \coresecond{81.74} & 68.90 & 43.27
& 80.95 & 73.49 & 67.39
& \coresecond{66.34} & \coresecond{53.98} \\

& MOMENTKV\corevenuetag{COLM'26}
& 93.29 & 68.50 & 49.80 & 33.48 & 21.17
& 80.95 & 56.59 & 38.80 & 26.91 & 20.39 \\

& \textbf{CORE (Ours)}
& 93.29
& \corebest{87.74}
& \corebest{84.91}
& \corebest{73.69}
& \corebest{54.63}
& 80.95
& \coresecond{75.37}
& \corebest{70.98}
& \corebest{68.98}
& \corebest{56.58} \\
\midrule

% ==================== Llama-3.1-8B ====================
\multirow{9}{*}{\rotatebox[origin=c]{90}{\textbf{Llama-3.1-8B}}}
& H$_2$O\corevenuetag{NeurIPS'23}
& 95.66 & 76.36 & 57.97 & 39.98 & 22.50
& 93.38 & 73.26 & 52.08 & 35.87 & 19.68 \\

& SnapKV\corevenuetag{NeurIPS'24}
& 95.66 & 81.62 & 69.10 & 50.41 & 24.60
& 93.38 & 82.32 & 69.25 & 48.56 & 26.99 \\

& PyramidKV\corevenuetag{COLM'25}
& 95.66 & 82.53 & 72.45 & 50.54 & 24.61
& 93.38 & 82.74 & 68.03 & 48.99 & 27.66 \\

& Expected Attention\corevenuetag{arXiv'25}
& 95.66 & \corebest{95.27} & 92.23 & 75.84 & 30.34
& 93.38 & \coresecond{92.84} & 85.93 & 66.31 & 25.25 \\

& KeyDiff\corevenuetag{NeurIPS'25}
& 95.66 & 91.59 & 85.38 & 72.95 & \coresecond{61.17}
& 93.38 & 88.42 & 82.55
& \coresecond{74.75} & \coresecond{66.68} \\

& CAPKV\corevenuetag{ICML'26}
& 95.66 & 92.73 & 84.29 & 67.08 & 50.53
& 93.38 & 91.20 & 84.26 & 67.68 & 50.06 \\

& IndexMem\corevenuetag{ICML'26}
& 95.66 & 93.43 & \coresecond{93.02}
& \coresecond{77.24} & 54.12
& 93.38 & 91.97 & \corebest{87.34} & 72.19 & 57.95 \\

& MOMENTKV\corevenuetag{COLM'26}
& 95.66 & 87.06 & 76.40 & 54.45 & 23.07
& 93.38 & 86.85 & 75.09 & 55.67 & 31.18 \\

& \textbf{CORE (Ours)}
& 95.66
& \coresecond{95.06}
& \corebest{93.77}
& \corebest{80.60}
& \corebest{62.44}
& 93.38
& \corebest{93.62}
& \coresecond{87.12}
& \corebest{74.98}
& \corebest{68.48} \\
\midrule

% ==================== Qwen3-14B ====================
\multirow{9}{*}{\rotatebox[origin=c]{90}{\textbf{Qwen3-14B}}}
& H$_2$O\corevenuetag{NeurIPS'23}
& 96.12 & 77.57 & 61.50 & 41.76 & 23.95
& 94.17 & 74.69 & 55.58 & 38.38 & 27.02 \\

& SnapKV\corevenuetag{NeurIPS'24}
& 96.12 & 86.82 & 61.61 & 41.42 & 29.23
& 94.17 & 85.33 & 67.79 & 44.29 & 32.43 \\

& PyramidKV\corevenuetag{COLM'25}
& 96.12 & 79.07 & 43.48 & 39.69 & 29.24
& 94.17 & 79.84 & 46.01 & 34.78 & 25.97 \\

& Expected Attention\corevenuetag{arXiv'25}
& 96.12
& \coresecond{96.07}
& \coresecond{95.88}
& \coresecond{86.80}
& \coresecond{58.80}
& 94.17
& \coresecond{94.21}
& \coresecond{93.90}
& \coresecond{83.08}
& 53.01 \\

& KeyDiff\corevenuetag{NeurIPS'25}
& 96.12 & 89.95 & 80.49 & 70.26 & 53.94
& 94.17 & 83.57 & 75.82 & 70.15 & \coresecond{62.24} \\

& CAPKV\corevenuetag{ICML'26}
& 96.12 & 89.94 & 83.41 & 69.58 & 45.35
& 94.17 & 89.32 & 83.68 & 73.11 & 53.79 \\

& IndexMem\corevenuetag{ICML'26}
& 96.12 & 95.60 & 93.59 & 84.73 & 52.49
& 94.17 & 91.98 & 86.39 & 77.14 & 52.48 \\

& MOMENTKV\corevenuetag{COLM'26}
& 96.12 & 90.22 & 69.32 & 43.04 & 27.63
& 94.17 & 88.89 & 73.67 & 54.40 & 33.90 \\

& \textbf{CORE (Ours)}
& 96.12
& \corebest{96.33}
& \corebest{95.96}
& \corebest{89.59}
& \corebest{60.40}
& 94.17
& \corebest{94.29}
& \corebest{93.96}
& \corebest{85.48}
& \corebest{66.02} \\
\bottomrule
\end{tabular}%
}
\vspace{-0.2in}
\end{table*}

\section{Experiments}
\label{sec:experiments}

We evaluate CORE’s long-context performance, decode-time compression, and position-sensitive retrieval, then analyze its retention and compensation mechanisms, ablate key components, compare with trainable KV eviction methods, and assess inference efficiency.

\subsection{Experimental Setup}
\label{sec:exp-setup}

\textbf{Models and benchmarks.}
We evaluate CORE on three backbones spanning 7B--14B parameters:
Mistral-7B~\citep{jiang2023mistral7b}, Llama-3.1-8B~\citep{grattafiori2024llama},
and Qwen3-14B~\citep{yang2025qwen3}. Long-context evaluation uses
RULER~\citep{hsieh2024ruler} and LongBench~\citep{bai2024longbench},
using each benchmark's official task definitions, metrics, and evaluation
procedures.

\textbf{Inference schedule and cache budget.}
We define the KV compression ratio $r\in[0,1]$ as the fraction of cached KV
entries evicted: 0\% denotes the full-KV reference, and increasing $r$
reduces the retained cache. All methods retain the first four sink tokens,
which count toward the shared KV budget. Each sequence is compressed once after prefill to the target budget; during decoding, the same fixed
KV budget is enforced every 128 generated tokens unless another interval is specified.

\textbf{Baselines.}
We compare CORE with H2O~\citep{zhang2023h2o}, SnapKV~\citep{li2024snapkv},
PyramidKV~\citep{cai2025pyramidkv}, ExpectedAttention~\citep{devoto2025expectedattention},
KeyDiff~\citep{park2026keydiff}, CAPKV~\citep{yang2026rethinking},
IndexMem~\citep{yang2026indexmem}, and MOMENTKV~\citep{li2026momentkvclosingdirectionalgap}. For IndexMem, we use the base architecture without
IndexCache-style cross-layer score reuse.
We evaluate all methods using a unified evaluation pipeline
based on the KVPress repository~\citep{devoto2025expectedattention},
with consistent benchmark settings and KV-cache budgets.
Unless otherwise specified, all experiments are conducted
on a single NVIDIA A100 GPU with 80\,GB of memory.

\textbf{Training Setup.} We train CORE on LongAlpaca-12k~\citep{chen2024longlora}
with a frozen backbone in two stages:
Stage I distills teacher allocations into the indexer using
KL divergence and boundary-focused pairwise supervision.
Stage II performs Top-$B$ retention, continues indexer distillation,
and trains per-layer latent-memory modules from evicted KV states.
Optimization uses AdamW with a warmup--stable--decay schedule
and peak learning rate $5\times10^{-4}$ over 5,400 total steps.

\subsection{Main Results}
\label{sec:main-results}

\textbf{RULER: Long-context performance.}
Table~\ref{tab:main-results} reports the mean over 13 RULER
subtasks at 4K and 16K context lengths.
CORE is strongest under aggressive compression, leading all 12
settings at 75\% and 90\% across three backbones and both lengths.
At 25\% and 50\% compression, it remains within 0.22 points
of the best baseline in every setting.
At 90\% compression, CORE improves over the strongest baseline
by 2.25 points on average, peaking at 3.78 points on
Qwen3-14B at 16K.

\textbf{LongBench: Performance across heterogeneous tasks.}
We test CORE's generalization beyond synthetic retrieval and reasoning
on LongBench~\citep{bai2024longbench} with Llama-3.1-8B.
Across six representative tasks (Fig.~\ref{fig:longbench-main}),
CORE maintains consistently strong, balanced performance,
with clearer advantages under aggressive KV-cache compression.
At 75\% compression, CORE's 16-task average is highest among
compressed methods, 2.30 points above the strongest baseline
(Appendix Table~\ref{tab:longbench-per-task}).

\begin{figure}[t]
\centering
\includegraphics[width=\linewidth]{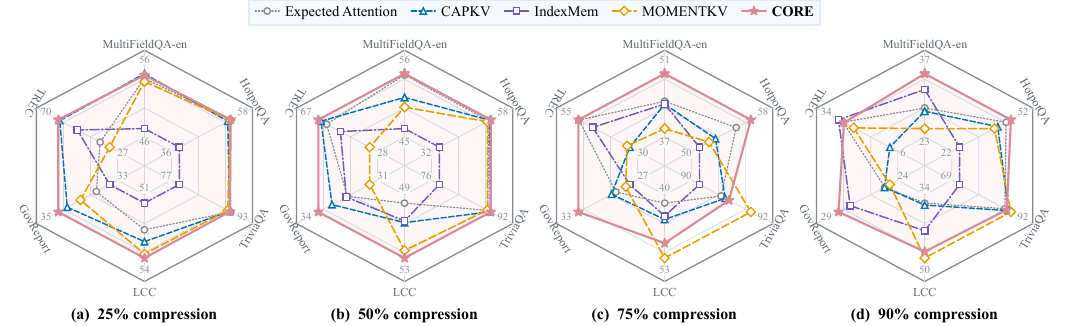}
\vspace{-0.3in}
\caption{
LongBench performance across six representative tasks and four compression ratios.
}
\label{fig:longbench-main}
\vspace{-0.15in}
\end{figure}

\subsection{Compression Robustness}
\label{sec:decode-compression}

\textbf{Needle retrieval across lengths and positions.}
Using Llama-3.1-8B at 50\% compression, we evaluate 2K--16K contexts
and needle depths of 10\%, 25\%, 50\%, 75\%, and 90\%,
with matched retained-KV budgets.
Fig.~\ref{fig:niah-heatmap} shows CORE maintaining high accuracy
across most length--depth pairs.
IndexMem exhibits recurring dips at 90\% depth, while CORE
without memory writing shows localized failures at intermediate
lengths and later positions in longer contexts.
Full CORE reduces these failure regions, supporting memory
compensation's contribution to more consistent retrieval.

\begin{wraptable}{r}{0.51\textwidth}

\centering
\scriptsize
\setlength{\tabcolsep}{2pt}
\renewcommand{\arraystretch}{0.85}
\vspace{-0.3in}
\caption{Decode-time compression results.}
\label{tab:decode-compression}

\begin{tabular}{lcccccccc}
\toprule
Method
& \multicolumn{4}{c}{AIME25}
& \multicolumn{4}{c}{Math500} \\
\cmidrule(lr){2-5}
\cmidrule(lr){6-9}
& 64 & 128 & 256 & 512
& 64 & 128 & 256 & 512 \\
\midrule
CAPKV
& 26.67 & 40.00 & 43.33 & 40.00
& 32.20 & 30.00 & 43.40 & 39.60 \\
IndexMem
& 23.33 & 33.33 & 43.33 & 46.67
& 35.20 & 40.60 & 47.40 & 50.20 \\
MOMENTKV
& 20.00 & 30.00 & 36.67 & 43.33
& 28.40 & 37.60 & 45.20 & 49.00 \\
\textbf{CORE}
& \textbf{33.33} & \textbf{43.33}
& \textbf{50.00} & \textbf{53.33}
& \textbf{37.60} & \textbf{43.00}
& \textbf{49.80} & \textbf{52.60} \\
\bottomrule
\end{tabular}

\vspace{-0.1in}
\end{wraptable}

\textbf{Decode-time compression.} We evaluate Qwen3-14B on AIME25 and Math500 with 75\% initial
eviction, matched fixed KV budgets, and recompression every
$\{64,128,256,512\}$ generated tokens.
CORE leads at every interval, reaching 53.33\% and 52.60\%
accuracy at interval 512.
Under the most frequent recompression schedule (interval 64),
CORE achieves 33.33\% and 37.60\%, exceeding MOMENTKV by
13.33 and 9.20 points.
It also achieves higher accuracy than IndexMem across all intervals
(Table~\ref{tab:decode-compression}),
demonstrating a consistent accuracy advantage.
These schedules vary both eviction frequency and intermediate
cache occupancy.

\begin{figure*}[t]
\centering
\includegraphics[width=\textwidth]{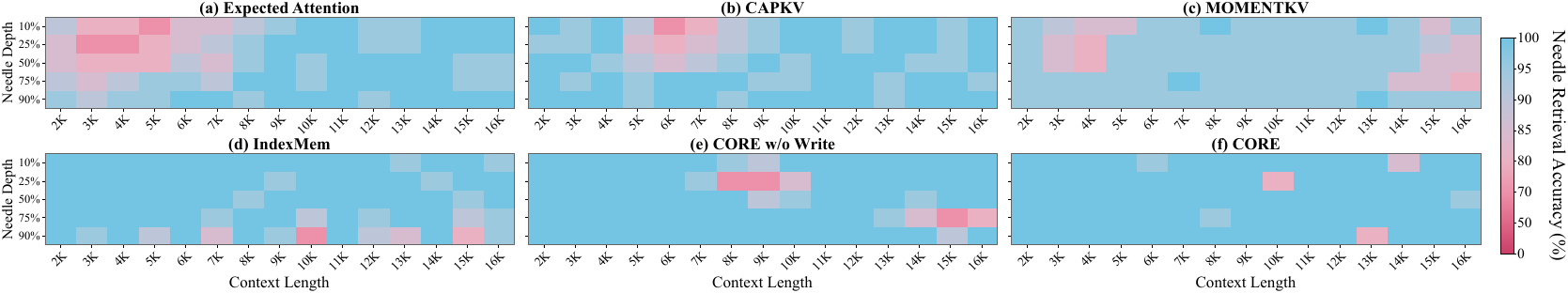}
\vspace{-0.3in}
\caption{Position-sensitive needle retrieval across context lengths
and evidence depths.}
\label{fig:niah-heatmap}
\vspace{-0.2in}
\end{figure*}

\subsection{Retention and Compensation Analysis}
\label{sec:mechanistic-analysis}

\textbf{Mass--direction retention.}
We evaluate Llama-3.1-8B on RULER-16K using frozen
modules and shared full-cache reference QKV states
(Appendix~\ref{app:diagnostic-protocol}).
With memory correction disabled,
Fig.~\ref{fig:mechanistic-analysis}~(a) reports mean
evicted attention mass $\bar p(E)$ and normalized directional
factor $\widetilde{\mathcal D}_E$ (lower is better).
Relative to the selector trained without the coverage objective,
CORE reduces both factors across the four budgets.
At high compression ratios, the masses become similar while CORE
maintains a smaller directional factor, indicating a contribution
beyond attention-mass retention.

\textbf{Two-channel allocation fidelity.}
CORE yields lower Top-$B$ teacher-mass regret and excluded-allocation
discrepancy than no-coverage-cues and KL-only students across all budgets
(Fig.~\ref{fig:mechanistic-analysis}~(b);
Appendix~\ref{app:joint-fidelity}).
Under the fixed-budget utility--coverage objective, Student
Top-$B$ remains close in objective value to Teacher Top-$B$ and greedy
selection (Appendix Table~\ref{tab:teacher-greedy}), showing
distillation largely preserves the teacher's coverage-informed set quality
and the value of set-level optimization.

\textbf{Mass-preserving write stability.}
We perturb
allocations to measure linear-write sensitivity.
Fig.~\ref{fig:mechanistic-analysis}~(c) shows mass-preserving
medians near 0.4 across excluded-mass bins,
while conditional-only sensitivity rises as $\mu_{\min}$ decreases;
90th-percentile trends are similar.
This accords with the mass-independent write bound and
conditional normalization's inverse-mass factor
(Proposition~\ref{prop:memory-stability}).
Appendix~\ref{app:compensation-diagnostics} evaluates
perturbation propagation after recurrent normalization.
\begin{figure*}[t]
\centering
\includegraphics[width=\textwidth]{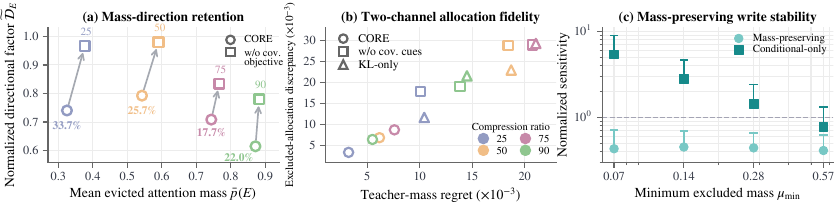}
\vspace{-0.35in}
\caption{Mechanistic analysis of CORE's retention, allocation fidelity, and write stability.}
\label{fig:mechanistic-analysis}
\vspace{-0.1in}
\end{figure*}

\begin{figure}[t]
\vspace{-0.0in}
    \centering
    \includegraphics[width=\linewidth]{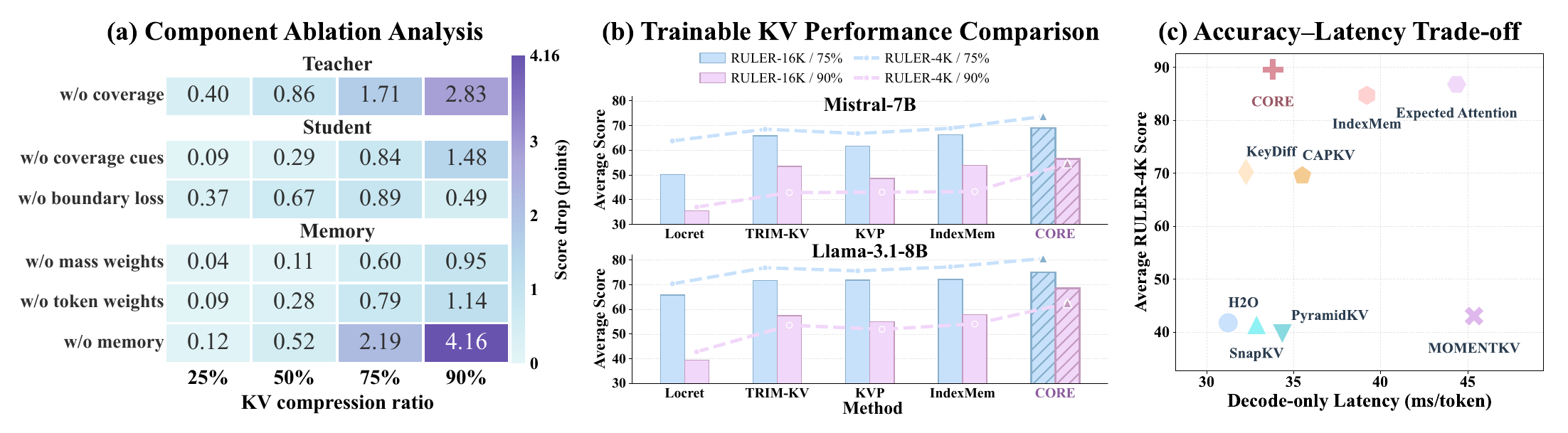}
\vspace{-0.3in}
    \caption{
       CORE Evaluation across Compression Performance, Component Analysis, and Efficiency.
    }
\vspace{-0.2in}
\label{fig:core_evaluation}
\end{figure}

\subsection{Ablations, Trainable Baselines, and Efficiency}
\label{sec:ablation}
\label{sec:efficiency}

\textbf{Coverage and memory improve compression robustness.}
Fig.~\ref{fig:core_evaluation}~(a) ablates CORE on
Llama-3.1-8B using RULER-16K.
Removing coverage supervision, coverage cues, or boundary loss
reduces performance.
Memory removal causes the largest drops at 75\%/90\% compression
(2.19 and 4.16 points), while removing mass or token weighting
degrades accuracy.
Mass- and token-weighting ablations follow matched two-stage
training, with Stage~II updating
the indexer and memory modules.

\textbf{CORE outperforms trainable eviction baselines.}
With all methods trained on LongAlpaca-12k,
CORE outperforms Locret~\citep{huang2025locret},
TRIM-KV~\citep{bui2026cache}, KVP~\citep{moschella2026learning},
and IndexMem~\citep{yang2026indexmem} across both backbones,
4K/16K contexts, and 75\%/90\% compression
(Fig.\ref{fig:core_evaluation}~(b)).
At 16K and 90\% compression, CORE exceeds IndexMem by 2.60 points
on Mistral-7B and 10.53 on Llama-3.1-8B, demonstrating its
advantage under tight KV budgets.

\textbf{CORE offers a favorable accuracy--latency trade-off.}
At 75\% compression on Qwen3-14B, CORE leads RULER-4K at
${\sim}34$\,ms/token, outperforming IndexMem and Expected Attention
(Fig.\ref{fig:core_evaluation}~(c)).
Decode latency is token-averaged suite-wide;
Appendix~\ref{app:efficiency-details} details efficiency and memory costs.

% \begin{wraptable}{r}{0.35\textwidth}
% \vspace{-2em}
% \centering
% \scriptsize
% \setlength{\tabcolsep}{3pt}
% \renewcommand{\arraystretch}{0.85}
% \caption{Allocation replacement.}
% \label{tab:allocation-ablation}

% \resizebox{\linewidth}{!}{%
% \begin{tabular}{@{}lcccc@{}}
% \toprule
% \multirow{2}{*}{\textbf{Variant}}
% & \multicolumn{2}{c}{\textbf{RULER-4K}}
% & \multicolumn{2}{c}{\textbf{RULER-16K}} \\
% \cmidrule(lr){2-3}
% \cmidrule(l){4-5}
% & 75\% & 90\% & 75\% & 90\% \\
% \midrule
% \textbf{CORE}
% & \textbf{80.60} & \textbf{62.44}
% & \textbf{74.98} & \textbf{68.48} \\
% $\text{CORE}_{\text{EA}}$
% & 78.22 & 50.39 & 70.65 & 56.87 \\
% $\text{CORE}_{\text{KeyDiff}}$
% & 76.78 & 61.53 & 74.87 & 67.16 \\
% \bottomrule
% \end{tabular}%
% }
% \vspace{-1.8em}
% \end{wraptable}

% \textbf{Allocation replacement.}
% On Llama-3.1-8B, $\text{CORE}_{\text{EA}}$ and $\text{CORE}_{\text{KeyDiff}}$ substitute normalized EA and KeyDiff scores for $\pi_\theta^{\mathcal C}$ in Top-$B$ retention and evicted-mass redistribution.
% We retrain memory maps and gates, matching data, initialization, optimization, and KV budgets.
% At 75\% and 90\% eviction, both outperform their baselines but trail CORE, supporting learned-allocation benefits beyond the compensation architecture.

\begin{wraptable}{r}{0.42\textwidth}
\vspace{-2em}
\centering
\scriptsize
\setlength{\tabcolsep}{3pt}
\renewcommand{\arraystretch}{0.8}
\caption{Write allocation comparison.}
\label{tab:shared-independent-allocation}
\resizebox{\linewidth}{!}{%
\begin{tabular}{@{}lcccc@{}}
\toprule
\multirow{2}{*}{\textbf{Write allocation}}
& \multicolumn{2}{c}{\textbf{RULER-4K}}
& \multicolumn{2}{c}{\textbf{RULER-16K}} \\
\cmidrule(lr){2-3}
\cmidrule(l){4-5}
& 75\% & 90\% & 75\% & 90\% \\
\midrule
CORE
& \textbf{80.60} & \textbf{62.44}
& \textbf{74.98} & \textbf{68.48} \\
$\text{CORE}_{\text{ind}}$
& 79.37 & 61.82
& 73.02 & 67.43 \\
\bottomrule
\end{tabular}%
}
\vspace{-1em}
\end{wraptable}

\textbf{Shared versus independent write allocations.}
On Llama-3.1-8B, we compare CORE's shared allocation with
$\mathrm{CORE}_{\mathrm{ind}}$, which learns a separate write
scorer via memory reconstruction.
CORE outperforms
$\mathrm{CORE}_{\mathrm{ind}}$ by 0.62--1.96 points across
4K/16K contexts and 75\%/90\% compression
(Table~\ref{tab:shared-independent-allocation}),
demonstrating the effectiveness of sharing the calibrated
allocation between retention and compensation.

\section{Conclusion}

CORE unifies retention and compensation through an exact factorization
of KV eviction error into discarded attention mass and directional gap.
It distills offline log-determinant coverage into a lightweight indexer
whose allocation drives Top-$B$ retention and parameterizes latent-memory
writes through excluded mass and conditional weights.
Our theory provides a four-term pre-compensation certificate,
mass-independent write stability, and hierarchical bounds through
recurrent and query-adaptive normalization.
Across long-context and repeated-eviction settings, CORE improves
accuracy--KV-budget trade-offs.
Future work could use excluded allocation mass to redistribute
a fixed KV budget across layers during decoding.

\section*{AI Use Statement}

We used generative AI tools primarily for language polishing and improving
the clarity, fluency, and readability of the manuscript. The tools were also
used for minor assistance with translation, wording refinement, terminology
consistency, and organization of technical descriptions. They occasionally
assisted in checking mathematical statements and proof arguments, as well as
clarifying the presentation of experimental settings, results, and related
technical content. These uses were limited to supporting manuscript preparation
and presentation.

\section*{Ethics Statement}
\label{app:ethics}

This work studies inference efficiency using existing pretrained
models and publicly available datasets and benchmarks. KV-cache
compression may discard information needed by later queries,
potentially causing retrieval failures or incorrect outputs.
These effects matter particularly in applications that require
reliable use of long-context evidence. Deployment should therefore
evaluate the chosen compression budgets on representative workloads
and monitor output quality, especially under repeated eviction
and changing query demands.

\section*{Reproducibility Statement}
\label{app:reproducibility}

Section~\ref{sec:exp-setup} describes the models, benchmarks,
baselines, hardware, and cache budgets.
Appendix~\ref{app:reporting-protocol} specifies implementation
details, training data, optimization settings, random seeds,
generation and scoring procedures, and hyperparameter selection.
Appendix~\ref{app:diagnostic-protocol} details protocols and metrics
for mechanistic analyses.
Appendix~\ref{app:additional-experiments} provides per-task results,
controlled comparisons, hyperparameter sensitivity, training-source
and seed sensitivity, and latency and memory measurements.

% \section*{Limitations}
% Although CORE performs well across three backbones, projected key geometry remains an indirect proxy for the information required by future queries.
% The stability bounds condition on a fixed query and do not cover
% autoregressive error propagation across future queries.
% While CORE achieves lower decoding latency, its feature extraction and initial cache preparation incur additional prefill cost; end-to-end speedups therefore depend on context length and how much subsequent generation amortizes this overhead.

% Author contributions are omitted from this anonymized manuscript.

% Acknowledgments are omitted from this anonymized manuscript.

\bibliography{conference}
\bibliographystyle{conference}

\clearpage
\appendix

{
\hypersetup{linkcolor=black}

\section*{\normalfont\large\bfseries Organization of the Appendix}

\normalsize
\setlength{\baselineskip}{14pt}
\setlength{\parskip}{5pt}

% A. Additional Method Details
\appendixentry{app:method-details}
  {Additional Method Details}

\appendixsubentry{app:method-positioning}
  {Positioning Relative to Closely Related KV-Cache Methods}

\appendixsubentry{app:teacher-solver}
  {Coverage Teacher and Fixed-Budget Approximation}

\appendixsubentry{app:feature-details}
  {Cache-Aware Indexer Specification}

\appendixsubentry{app:boundary-loss}
  {Budget-Conditioned Boundary Supervision}

% B. Theoretical Details
\appendixentry{app:theory-proofs}
  {Theoretical Details for Sec.~\ref*{sec:theory}}

\appendixsubentry{app:scope}
  {Scope and Characterization of the Identity}

\appendixsubentry{app:eviction-certificates}
  {Eviction Identity and Certificates}

\appendixsubentry{app:allocation-retention-coupling}
  {Allocation, Retention Fidelity, and Two-Channel Coupling}

\appendixsubentry{app:proof-cache-conditioned-coverage}
  {Coverage Interaction and Set Awareness}

\appendixsubentry{app:mass-preserving-recurrent-memory}
  {Mass-Preserving and Recurrent Memory}

% C. Theory-Aligned Empirical Diagnostics
\appendixentry{app:diagnostic-protocol}
  {Theory-Aligned Empirical Diagnostics for Sec.~\ref*{sec:mechanistic-analysis}}

\appendixsubentry{app:retention-fidelity}
  {Retention and Allocation Fidelity}

\appendixsubentry{app:compensation-diagnostics}
  {Memory Reconstruction and Stability}

% D. Implementation and Evaluation Protocol
\appendixentry{app:reporting-protocol}
  {Implementation and Evaluation Protocol}

\appendixsubentry{app:implementation-details}
  {Implementation Details}

\appendixsubentry{app:d.2}
  {Training Details}

\appendixsubentry{app:d.3}
  {Evaluation Protocol}

\appendixsubentry{app:validation-selection}
  {Development Set and Hyperparameter Selection}

% E. Additional Experimental Results
\appendixentry{app:additional-experiments}
  {Additional Experimental Results}

\appendixsubentry{app:detailed-results}
  {Per-task Results on LongBench and RULER at 75\% Compression}

\appendixsubentry{app:ultra-long}
  {Ultra-Long-Context Evaluation}

\appendixsubentry{app:hyperparameter-sensitivity}
  {Hyperparameter Sensitivity}

\appendixsubentry{app:shared-alternative-allocation}
  {Shared and alternative allocation strategies}

\appendixsubentry{app:training-sensitivity}
  {Sensitivity to Training Data and Randomness}

\appendixsubentry{app:efficiency-details}
  {Full Efficiency Scaling and Parameter Cost}

  \appendixentry{app:statements}
  {Limitations}
}
\par\vspace{0.8em}
\noindent
Appendix~\ref{app:method-details} positions CORE among related
methods and details the coverage teacher, cache-aware indexer,
and budget-conditioned distillation.
Appendix~\ref{app:theory-proofs} provides proofs and analyses of the
eviction identity, error certificates, allocation fidelity,
coverage interactions, and recurrent-memory stability.
Appendix~\ref{app:diagnostic-protocol} presents empirical diagnostics
of retention, allocation fidelity, and memory reconstruction.
Appendix~\ref{app:reporting-protocol} documents implementation,
training, evaluation, and hyperparameter selection.
Appendix~\ref{app:additional-experiments} reports
per-task and ultra-long-context results, allocation comparisons,
sensitivity analyses, and detailed efficiency measurements.
Appendix~\ref{app:statements} discusses CORE's limitations.

\clearpage

\appendix

\section{Additional Method Details}
\label{app:method-details}

\subsection{Positioning Relative to Closely Related KV-Cache Methods}
\label{app:method-positioning}

Table~\ref{tab:method-positioning} summarizes the distinction from the related methods. The mass--direction decomposition provides a common analytical
view of KV eviction and is also used by MOMENTKV.
CORE builds on this perspective by using one coverage-calibrated
allocation to connect retention and compensation.

\begin{table}[htp]
\centering
\vspace{-0.2in}
\caption{
Conceptual positioning of closely related fixed-budget KV-cache methods.
}
\label{tab:method-positioning}
\small
\setlength{\tabcolsep}{3.2pt}
\renewcommand{\arraystretch}{0.8}
\resizebox{\columnwidth}{!}{
\begin{tabular}{lllll}
\toprule
\textbf{Method}
& \textbf{Retention principle}
& \textbf{Evicted-state representation}
& \textbf{Compensation}
& \textbf{Retention--compensation link} \\
\midrule

CAPKV
& Information-capacity selection
& None
& None
& Retention only \\

MOMENTKV
& Moment-aware eviction
& Low-order moments
& First-order approximation
& Shared moment statistics \\

IndexMem
& Learned importance indexer
& Latent memory
& Gated residual readout
& Separate modules \\

\textbf{CORE}
& Coverage-calibrated allocation
& Allocation-weighted latent memory
& Gated residual reconstruction
& Shared allocation \\

\bottomrule
\end{tabular}}
\vspace{-0.1in}
\end{table}

CORE shares allocation weights across the two channels:
Top-$B$ determines retained support, while
$\mu_E\omega_i=\pi_\theta^{\mathcal C}(i)$ determines
evicted-state writes.
The memory map and gate remain learned components,
but no separate write-weight predictor is required.

\subsection{Coverage Teacher and Fixed-Budget Approximation}
\label{app:teacher-solver}

The teacher transfers non-additive coverage information into a distribution that can be distilled into a lightweight indexer. We distinguish the continuous optimum, the finite distillation
target, and its induced Top-$B$ set, and establish the guarantees
relating them.

\subsubsection{Coverage Marginals and Teacher Optimality}
\label{app:teacher-coverage}
\label{app:teacher-concavity}

For the cache-centered projected keys $z_i^c=z_i-\frac{1}{N}\sum_j z_j$, define
\begin{equation}
\Sigma(w)=I+\alpha\sum\nolimits_{i=1}^{N}w_i z_i^c(z_i^c)^\top,
\qquad
F(w)=\log\det\Sigma(w).
\label{eq:app-teacher-coverage}
\end{equation}
Jacobi's identity gives the marginal coverage signal
$d_i(w)=\frac{\partial F(w)}{\partial w_i}
=\alpha(z_i^c)^\top\Sigma(w)^{-1}z_i^c.$
The marginal depends on the current allocation-weighted cache
geometry: for equal vector norms, directions less represented
in $\Sigma(w)$ receive larger scores. It therefore captures
cache-conditioned coverage rather than an intrinsic token score.

\begin{proposition}[Teacher objective well-posedness]
\label{prop:teacher-optimality}
For $\alpha>0$, $\beta\geq0$, and $\tau_{\mathrm{ent}}>0$, the objective
in Eq.~\ref{eq:teacher-objective} is strictly concave on $\Delta_N$.
It has a unique maximizer $w^\star\in\operatorname{ri}(\Delta_N)$
satisfying
\begin{equation}
w^\star=\operatorname{softmax}\!\left(\frac{u+\beta d(w^\star)}{\tau_{\mathrm{ent}}}\right).
\label{eq:app-teacher-fixed-point}
\end{equation}
\end{proposition}

\begin{proof}
The utility term is affine. Since $\Sigma(w)\succeq I\succ0$, the
log-determinant term is concave in $w$, while entropy is strictly concave
for $\tau_{\mathrm{ent}}>0$. Continuity on the compact simplex gives
existence, and strict concavity gives uniqueness. The maximizer is interior:
moving mass $\epsilon$ from a positive coordinate to a zero coordinate
yields an entropy gain of order $-\epsilon\log\epsilon$, dominating the
$O(\epsilon)$ change in utility and coverage. At the interior optimum,
stationarity under the simplex constraint requires
\begin{equation}
u_i+\beta d_i(w^\star)-\tau_{\mathrm{ent}}(\log w_i^\star+1)-\nu=0.
\label{eq:app-teacher-stationarity-condition}
\end{equation}
Exponentiating and normalizing gives Eq.~\ref{eq:app-teacher-fixed-point}.
\end{proof}

\paragraph{Coverage-weight monotonicity.}
Write $J_\beta(w)=G(w)+\beta F(w)$, where
$G(w)=\langle u,w\rangle
-\tau_{\mathrm{ent}}\sum_i w_i\log w_i$.
For $0\leq\beta_1<\beta_2$, adding the optimality inequalities
for $w_j=w_{\beta_j}^\star$ gives
$(\beta_2-\beta_1)\bigl[F(w_2)-F(w_1)\bigr]\geq0.$
Thus the exact optimum's log-determinant coverage is
nondecreasing in $\beta$, without implying a monotonic
reduction in the output-space directional factor $\mathcal D_E$.

\subsubsection{Finite-Step Distillation Target}
\label{app:teacher-update-derivation}
\label{app:teacher-stopping}

Starting from $w_i^{(0)}=1/N$, CORE applies
Eq.~\ref{eq:teacher-update} for $R$ steps and distills
$\pi_T^{\mathcal C}=w^{(R)}$.
Holding $d^{(r)}=d(w^{(r)})$ fixed, this update uniquely solves the entropy-regularized linear surrogate
\begin{equation}
w^{(r+1)}
=\arg\max_{w\in\Delta_N}
\left\{
\sum\nolimits_{i=1}^{N}w_i\bigl(u_i+\beta d_i^{(r)}\bigr)
-\tau_{\mathrm{ent}}\sum\nolimits_{i=1}^{N}w_i\log w_i
\right\}.
\label{eq:app-teacher-linearized}
\end{equation}
If the full sequence of iterates converges, its limit satisfies Eq.~\ref{eq:app-teacher-fixed-point} and hence equals the unique optimizer. Strict concavity does not guarantee global convergence or monotonic objective improvement for this iteration. Distillation uses the finite target $\pi_T^{\mathcal C}$, while online inference uses the trained indexer.

\begin{proposition}[Finite-target certificate]
\label{prop:finite-target-certificate}
Let $J$ be the objective in
Eq.~\ref{eq:teacher-objective}; define the fixed-point map
\[
\mathcal T(w)
=
\operatorname{softmax}
\left(
\frac{u+\beta d(w)}{\tau_{\mathrm{ent}}}
\right),
\qquad
e(w)
=
D_{\mathrm{KL}}\!\left(w\,\middle\|\,\mathcal T(w)\right).
\]
For every $w\in\operatorname{ri}(\Delta_N)$,
$0
\leq
J(w^\star)-J(w)
\leq
\tau_{\mathrm{ent}} e(w),
D_{\mathrm{KL}}(w\|w^\star)
\leq
e(w).$
In particular, let
$e_R=e(w^{(R)}),
\gamma_{B,R}
=
w_{(B)}^{(R)}-w_{(B+1)}^{(R)},$
with coordinates sorted in nonincreasing order. Then
\[
\gamma_{B,R}>\sqrt{2e_R}
\quad\Longrightarrow\quad
\operatorname{TopB}(w^{(R)})
=
\operatorname{TopB}(w^\star).
\]
With four protected sinks and $4<B<N$, the same criterion
applies to protected retained sets using the cutoff gap at
rank $B-4$ among non-sink coordinates of $w^{(R)}$,
without renormalization.
\end{proposition}

\begin{proof}
Write $H(z)=-\sum_i z_i\log z_i$. For fixed $w$, define the entropy-regularized linearization
\[
Q_w(z)
=
\langle u,z\rangle
+
\beta
\left[
F(w)+\langle d(w),z-w\rangle
\right]
+
\tau_{\mathrm{ent}}H(z).
\]
Concavity of $F$ gives $J(z)\leq Q_w(z)$, with equality at $z=w$,
and $\mathcal T(w)$ is the unique maximizer of $Q_w$. Therefore,
$J(w^\star)-J(w)
\leq
Q_w(\mathcal T(w))-Q_w(w)
=
\tau_{\mathrm{ent}}
D_{\mathrm{KL}}\!\left(w\,\middle\|\,\mathcal T(w)\right).$
Concavity of $F$ together with the stationarity of $w^\star$ also gives
$J(w^\star)-J(w)
\geq
\tau_{\mathrm{ent}}D_{\mathrm{KL}}(w\|w^\star).$
Combining the two inequalities yields
$D_{\mathrm{KL}}(w\|w^\star)\leq e(w)$.
Pinsker's inequality gives
$\|w^{(R)}-w^\star\|_1\leq\sqrt{2e_R}$.
A cutoff gap exceeding this bound prevents any
retained--excluded pair from reversing order.
\end{proof}

\paragraph{Finite-step numerical check.}

\begin{wraptable}{r}{0.45\textwidth}
\vspace{-0.33in}
\centering\scriptsize
\setlength{\tabcolsep}{1.5pt}
\renewcommand{\arraystretch}{0.8}
\caption{Teacher residuals and Top-$B$ selection agreement
(relative to $R=16$).}
\label{tab:teacher-certificates}
\begin{tabular}{lrrr}
\toprule
Steps & Median $e_R$ & P90 $e_R$ & Overlap vs.\ 16\\
\midrule
$1$ & $7.643\times10^{-3}$ & $2.407\times10^{-2}$ & 92.80\%\\
$R_0=8$ & $2.176\times10^{-7}$ & $5.564\times10^{-5}$ & 99.74\%\\
$2R_0=16$ & $4.025\times10^{-12}$ & $1.496\times10^{-7}$ & ref.\\
\bottomrule
\end{tabular}
\vspace{-0.17in}
\end{wraptable}

We evaluate 104 RULER-16K sequences (eight per task)
at four layers spanning model depth.
Table~\ref{tab:teacher-certificates} reports median and
P90 residuals $e_R$ over valid sequence--layer states.
Selection overlap $|S_R(B)\cap S_{16}(B)|/B$, where
$S_R(B)=\operatorname{TopB}(w^{(R)})$, is averaged across
these states and compression ratios $\{0.25,0.50,0.75,0.90\}$
under the common selection protocol.
Results at $R_0=8$ indicate small residuals and stable selections.
Overlap with the $R=16$ reference alone does not certify exact
Top-$B$ recovery; Proposition~\ref{prop:finite-target-certificate}
provides a sufficient condition.

\subsubsection{Connection to Fixed-Budget Selection}
\label{app:fixed-budget-relation}

The continuous objective is connected to set selection by restricting it to uniform allocations on size-$B$ subsets. For $S\subseteq\mathcal C$ with $|S|=B$, let
$w_i^S=\frac{1}{B}\mathbb I[i\in S].$
Then
\begin{equation}
J(w^S)=J_B(S)+\tau_{\mathrm{ent}}\log B,
\label{eq:app-budget-restriction}
\end{equation}
\begin{equation}
J_B(S)=\frac{1}{B}\sum\nolimits_{i\in S}u_i
+\beta\log\det\!\left(I+\frac{\alpha}{B}\sum\nolimits_{i\in S}z_i^c(z_i^c)^\top\right).
\label{eq:app-fixed-budget-objective}
\end{equation}
The entropy contribution is constant at fixed $B$, so this restriction preserves the ordering of the non-additive set objective $J_B$.
The finite teacher ranking has a different, modular interpretation. For $R\geq1$, monotonicity of softmax gives
\begin{equation}
\operatorname{TopB}(\pi_T^{\mathcal C})
\in\arg\max\nolimits_{\substack{S\subseteq\mathcal C\\|S|=B}}
\sum\nolimits_{i\in S}\left[u_i+\beta d_i(w^{(R-1)})\right].
\label{eq:app-modular-topb}
\end{equation}
A fixed tie-breaking rule selects among equivalent maximizers. The ranking therefore optimizes a surrogate whose coefficients incorporate the final cache-conditioned coverage marginals. It need not maximize $J_B(S)$: the marginals are evaluated at a common allocation rather than recomputed after each discrete selection. Appendix~\ref{app:proof-cache-conditioned-coverage} characterizes the coverage interactions omitted by this scalarization. The student further amortizes this teacher ranking through distillation.

\subsection{Cache-Aware Indexer Specification}
\label{app:feature-details}

Each candidate receives a 14-dimensional input
$\xi_i=[\xi_i^{\mathrm{rel}};\xi_i^{\mathrm{cov}};
\xi_i^{\mathrm{aux}}]$:
six relation features summarize observed query demand,
four coverage features describe cache geometry, and four
auxiliary features encode position, value, and execution state.

\subsubsection{Query--Key Relation Features}
\label{app:relation-features}

Student features use the cache-relative responses $r_{t,h,i}$
and group aggregates $R_{t,m,i}$ in Eq.~\ref{eq:relation-signal}.
Ignoring the stabilizer, $r_{t,h,i}>0$ indicates attention above
the row's uniform baseline.

The relation vector is $\xi_i^{\mathrm{rel}}=[f_1(i),\ldots,f_6(i)]$, with
\begin{equation}
\begin{aligned}
f_1(i)&=\max\nolimits_{t,m}R_{t,m,i},
&
f_2(i)&=\operatorname{TopMean}_{t,m}R_{t,m,i},
\\
f_3(i)&=\operatorname{Mean}_{(t,m):R_{t,m,i}>0}R_{t,m,i},
&
f_4(i)&=\operatorname{TopMean}_{t\in\mathcal Q_{\mathrm{recent}},m}R_{t,m,i},
\\
f_5(i)&=\frac{|\{t:\max_m R_{t,m,i}>0\}|}{|\mathcal Q|},
&
f_6(i)&=\frac{|\{m:\operatorname{TopMean}_{t\in\mathcal Q}R_{t,m,i}>0\}|}{H_{\mathrm{kv}}}.
\end{aligned}
\label{eq:app-relation-feature-vector}
\end{equation}
These features capture peak response, the strongest responses, mean positive response, recent-query demand, query-level support, and KV-group support. We set $f_3(i)=0$ when its positive set is empty. Here $\operatorname{TopMean}$ denotes the mean of the largest entries under
the aggregation size specified for each use. We use $K=4$ by
default, while $f_6$ uses the top 25\% of causally visible queries. Causally invisible query--candidate pairs are masked; they do not contribute positive support.
The statistics can be accumulated over query rows without
storing the full query--candidate matrix; computing the
underlying QK responses remains part of online feature extraction.

\subsubsection{Coverage and Redundancy Features}
\label{app:coverage-features}

Teacher and student share the projected and cache-centered geometry
in Eq.~\ref{eq:coverage-representation}.
The entries of the layer-specific matrix
$P_k^{(\ell)}\in\mathbb R^{d_z\times H_{\mathrm{kv}}d_{\mathrm{head}}}$
are sampled from $\mathcal N(0,1/d_z)$ once and held fixed.

\textbf{Local directional descriptors.}
Let
$\widehat z_i^c=\frac{z_i^c}{\max(\|z_i^c\|_2,\epsilon_{\mathrm{geo}})},
\epsilon_{\mathrm{geo}}=10^{-6}.$
Reliable unit directional descriptors require
$\|z_i^c\|_2>\epsilon_{\mathrm{geo}}$.
We sample $M_p\geq2$ fixed unit prototypes as
\begin{equation}
\widetilde p_m\sim\mathcal N(0,I_{d_z}),
\qquad
p_m=\widetilde p_m/\|\widetilde p_m\|_2.
\label{eq:app-random-prototypes}
\end{equation}
The random seed is fixed, and the prototypes are unchanged during training and inference.

For reliable candidates, $m_1(i)$ and $m_2(i)$ index the largest and second-largest similarities $(\widehat z_i^c)^\top p_m$, respectively. Define the occupancy of each prototype using reliable candidates only:
{\small
\begin{equation}
n_m=\left|\left\{j:\|z_j^c\|_2>\epsilon_{\mathrm{geo}},\ m_1(j)=m\right\}\right|.
\label{eq:app-prototype-count}
\end{equation}
The local descriptors are
\begin{equation}
\left(o_i,d_i^{\mathrm{proto}},\Delta_i^{\mathrm{proto}}\right)
=
\begin{cases}
\left(
\log(1+n_{m_1(i)}),
1-(\widehat z_i^c)^\top p_{m_1(i)},
(\widehat z_i^c)^\top[p_{m_1(i)}-p_{m_2(i)}]
\right),
& \|z_i^c\|_2>\epsilon_{\mathrm{geo}},
\\
(0,1,0),
& \|z_i^c\|_2\leq\epsilon_{\mathrm{geo}}.
\end{cases}
\label{eq:app-prototype-descriptors}
\end{equation}}
They measure directional crowding, distance from the assigned
prototype, and assignment margin.

\textbf{Global cache-relative deviation.}
To complement local occupancy, apply a fixed randomized Hadamard rotation
\begin{equation}
\widetilde z_i=Rz_i,
\qquad
R=\frac{1}{\sqrt{d_z}}HD,
\label{eq:app-hadamard-rotation}
\end{equation}
where $D$ is a random-sign diagonal matrix and $HH^\top=d_z I$. The rotation preserves Euclidean geometry, with centered coordinates $\widetilde z_i-\frac{1}{N}\sum_j\widetilde z_j=Rz_i^c$. The subsequent coordinate-wise standardization is defined in this fixed rotated basis.
For coordinate $r$, compute
\begin{equation}
\mu_r=\frac{1}{N}\sum\nolimits_j\widetilde z_{j,r},
\qquad
\sigma_r^2=\frac{1}{N}\sum\nolimits_j\widetilde z_{j,r}^{\,2}-\mu_r^2,
\qquad
\sigma_{r,\mathrm{eff}}^2=(1-\rho_{\mathrm{sh}})\sigma_r^2
+\frac{\rho_{\mathrm{sh}}}{d_z}\sum\nolimits_{a=1}^{d_z}\sigma_a^2,
\label{eq:app-variance-shrinkage}
\end{equation}
where $\rho_{\mathrm{sh}}\in[0,1]$. The global descriptor is
\begin{equation}
g_i^{\mathrm{dev}}
=\log\!\left[
1+\frac{1}{d_z}\sum\nolimits_{r=1}^{d_z}
\frac{(\widetilde z_{i,r}-\mu_r)^2}{\sigma_{r,\mathrm{eff}}^2+\epsilon}
\right].
\label{eq:app-global-deviation}
\end{equation}
The complete coverage vector is
$\xi_i^{\mathrm{cov}}=[o_i,d_i^{\mathrm{proto}},\Delta_i^{\mathrm{proto}},g_i^{\mathrm{dev}}].$
These descriptors provide compact cache-relative cues. They do not reconstruct $\Sigma(w)^{-1}$ or analytically recover the teacher marginal; their relationship to the coverage-calibrated allocation is learned through distillation.

\subsubsection{Auxiliary Cache-State Features}
\label{app:auxiliary-features}

Four auxiliary descriptors (normalized age, sink/special-token status, value energy, execution phase):
\begin{equation}
\begin{aligned}
a_i^{\mathrm{age}}
&=\frac{\log(1+t_{\mathrm{cur}}-\operatorname{pos}_i)}{\log(1+t_{\mathrm{cur}})},\quad
a_i^{\mathrm{sink}}
=\mathbb I[\operatorname{pos}_i\leq N_{\mathrm{sink}}]
\vee\mathbb I[\operatorname{tok}_i\in V_{\mathrm{special}}],
\\
a_i^{\mathrm{value}}
&=\log\!\left(1+\frac{1}{H_{\mathrm{kv}}}\sum\nolimits_{m=1}^{H_{\mathrm{kv}}}
\frac{\|v_i^{(m)}\|_2^2}{d_{\mathrm{head}}}\right),
\quad
a_i^{\mathrm{phase}}
=\mathbb I[\mathrm{decode}].
\end{aligned}
\label{eq:app-auxiliary-descriptors}
\end{equation}
Thus
$\xi_i^{\mathrm{aux}}
=[a_i^{\mathrm{age}},a_i^{\mathrm{sink}},a_i^{\mathrm{value}},a_i^{\mathrm{phase}}].$
The sink indicator supplies a feature; any mandatory protection is enforced separately by the retention policy and counted within the retained-KV budget.

\subsection{Budget-Conditioned Boundary Supervision}
\label{app:boundary-loss}

Protected positions are excluded from ranking supervision.
Below, $N$ denotes the number of eligible candidates and $B$
the remaining retention budget.
Let $\operatorname{rank}_T(i)$ rank these candidates by
decreasing teacher probability, with fixed tie-breaking.
For boundary width $\Delta_B$, define
\begin{align}
P_{\mathrm{bd}}
&=\{i:\max(1,B-\Delta_B+1)\leq\operatorname{rank}_T(i)\leq B\},
\label{eq:app-positive-boundary}\\
N_{\mathrm{bd}}
&=\{j:B+1\leq\operatorname{rank}_T(j)\leq\min(N,B+\Delta_B)\}.
\label{eq:app-negative-boundary}
\end{align}
These windows contain candidates immediately above and below
the teacher cutoff.
For nonempty windows, the pairwise loss is in Eq.~\ref{eq:indexer-loss}.
It penalizes inversions and correctly ordered pairs with
separation below $m_{\mathrm{bd}}$, but local ordering alone
does not guarantee complete teacher--student Top-$B$ agreement.
One indexer supports multiple budgets using the same teacher
allocation and ranking; only the boundary windows vary.
Prefill supervision averages the boundary loss over all
configured budgets, while single-budget decode supervision
uses a uniformly sampled budget
(Appendix~\ref{app:d.2}).

\section{Theoretical Details for Sec.~\ref{sec:theory}}
\label{app:theory-proofs}

\subsection{Scope and Characterization of the Identity}
\label{app:scope}

We prove Theorem~\ref{thm:characterization} by separating the
algebraic factorization from its converse.
Throughout, $\mathcal I$ is finite with $|\mathcal I|\geq3$,
and $o:2^{\mathcal I}\setminus\{\emptyset\}\to\mathbb R^d$
with $d\geq2$.

% ----------------------------------------------------------------------
\subsubsection{Setup and Sufficiency}
\label{app:scope-setup}

A normalized aggregation readout has positive weights $w_i>0$
and vectors $v_i\in\mathbb R^d$ such that, for every nonempty
$S\subseteq\mathcal I$,
\begin{equation}
o(S)=\frac{N(S)}{Z(S)},\qquad
N(S)=\sum\nolimits_{i\in S}w_iv_i,\qquad
Z(S)=\sum\nolimits_{i\in S}w_i.
\label{eq:app-nar}
\end{equation}
For disjoint nonempty $M,E$, additivity gives
\[
o(M\cup E)=(1-\pi_E)o(M)+\pi_Eo(E),\qquad
\pi_E=\frac{Z(E)}{Z(M)+Z(E)}\in(0,1).
\]
Subtracting $o(M)$ proves the factorization in
Theorem~\ref{thm:characterization}, without a non-degeneracy
assumption.

% ----------------------------------------------------------------------
\subsubsection{Necessity: Proof of Theorem~\ref{thm:characterization}}
\label{app:scope-characterization}

\begin{theorem}[Characterization under non-degeneracy]
\label{thm:characterization-formal}
Assume the following conditions:
\begin{enumerate}[label=(A\arabic*),leftmargin=2.4em,itemsep=0pt]
\item for all disjoint nonempty $M,E$ there exist
$\lambda(M,E)\in(0,1)$ and $g(E)\in\mathbb R^d$, the latter depending
only on $E$, with
$
o(M\cup E)
=
(1-\lambda(M,E))o(M)+\lambda(M,E)g(E);
$
\item $g(\{i\})=o(\{i\})$ for all $i$;
\item the readouts of any three pairwise disjoint nonempty sets are
affinely independent, and singleton readouts are pairwise distinct.
\end{enumerate}
Under these conditions there are positive weights $w_i$, unique up to a
common positive scale, for which, with $v_i:=o(\{i\})$ and
$W(S):=\sum_{i\in S}w_i$,
$
o(S)
=
\frac{\sum_{i\in S}w_iv_i}{W(S)},
$
for every nonempty $S$. Conversely, every readout in
Eq.~\ref{eq:app-nar} satisfies (A1)--(A2) with
$g(E)=o(E)$ and
$\lambda(M,E)=\frac{W(E)}{W(M)+W(E)}$.
\end{theorem}

\begin{proof}
Write
$\mu(M,a):=\lambda(M,\{a\})$ and
$v_a:=o(\{a\})$.

\emph{Exchange identity and odds recursion.}
Under (A1)--(A2),
$
o(M\cup\{a\})
=
(1-\mu(M,a))o(M)+\mu(M,a)v_a.
$
Compute $o(M\cup\{a,b\})$ by inserting $a$ then $b$ and vice versa.
Write $A=\mu(M,a)$ and $B=\mu(M,b)$, with
$A'=\mu(M\cup\{b\},a)$ and $B'=\mu(M\cup\{a\},b)$.
By (A3), matching coefficients of the affinely independent triple
$o(M),v_a,v_b$ gives $A'=(1-B')A$ and $B'=(1-A')B$.
Thus $A'=A(1-B)/(1-AB)$.
Defining the insertion odds
$R(M,a):=\mu(M,a)/(1-\mu(M,a))$ yields
\begin{equation}
\vspace{-0.05in}
R(M\cup\{b\},a)
=
\frac{R(M,a)}{1+R(M,b)}.
\label{eq:app-odds}
\vspace{-0.05in}
\end{equation}

\emph{Pair weights.}
Fix $c_0$, set $w_{c_0}:=1$, and
$
w_a:=R(\{c_0\},a).
$
For distinct $c,a\neq c_0$, expanding $o(\{c_0,c,a\})$ along
$M=\{c_0,c\}$ using Eq.~\ref{eq:app-odds} gives
\vspace{-0.05in}
\begin{equation}
o(\{c_0,c,a\})
=
\frac{
v_{c_0}+w_cv_c+w_av_a
}{
1+w_c+w_a
}.
\label{eq:app-triple}
\vspace{-0.05in}
\end{equation}
Expanding instead along $M=\{c,a\}$ with $c_0$ inserted last and matching coefficients of $v_{c_0}$, $v_c$, and $v_a$ against Eq.~\ref{eq:app-triple} gives $\mu(\{c\},a)=\frac{w_a}{w_c+w_a}$, so $R(\{c\},a)=w_a/w_c$.
For pairs containing $c_0$, the same ratio follows from the
definition of $w_a$ and the complementary coefficients in
$o(\{c_0,a\})$.

\emph{Ratio form.}
The odds recursion extends the pairwise ratios to arbitrary subsets.
Induction on $|M|$ in Eq.~\ref{eq:app-odds} gives $R(M,a)=w_a/W(M)$, and therefore $\mu(M,a)=w_a/(W(M)+w_a)$. A second induction on $|S|$ through the insertion rule then gives $o(S)=\frac{\sum_{i\in S}w_i v_i}{W(S)}$ for every nonempty $S$.

\emph{Uniqueness and canonical factors.}
Singleton readouts are distinct, so each two-item readout uniquely
identifies its convex mixing coefficient. The pair-weight construction
therefore fixes every ratio $w_a/w_{c_0}$, proving uniqueness up to a
common positive scale. Appendix~\ref{app:scope-setup} then gives
the canonical factors $g(E)=o(E)$ and
$\lambda(M,E)=W(E)/(W(M)+W(E))$.
\end{proof}

\textbf{Scope of identification.}
Assumption (A3) enables affine coefficient matching in the converse;
the sufficient factorization requires no non-degeneracy.
The converse identifies the normalized representation, not every
auxiliary $g$ allowed by (A1). In particular, $g(E)$ need not be
determined when $|\mathcal I\setminus E|=1$.

\subsubsection{Boundary Cases}
\label{app:scope-boundary}

\begin{remark}[Boundary cases]
The characterization also delineates three nearby regimes in which the
interpretation changes:

\emph{(i) Signed weights.}
Retaining (A2)--(A3), allow
$\lambda\in\mathbb R\setminus\{0,1\}$ in (A1).
The same coefficient matching and odds recursion yield
$o(S)=\frac{N(S)}{W(S)}$ with nonzero signed weights;
$R(M,a)\neq-1$ ensures nonzero subset sums inductively.
The factorization applies to signed normalized linear attention,
but $\pi_E$ may leave $[0,1]$, and neither the probability
interpretation nor the bound $\|\Delta o\|_2\leq2B_v\pi_E$
is guaranteed.

\emph{(ii) Unnormalized additive readouts.}
If $o(S)=N(S)$ with additive $N$, then
$o(M\cup E)-o(M)=N(E)$; no subset renormalization is involved.

\emph{(iii) Non-subset compression.}
Compression that constructs new representations, such as soft prompts
or compressed context tokens, lies outside the subset operation in (A1).
\end{remark}

\subsection{Eviction Identity and Certificates}
\label{app:eviction-certificates}

% ----------------------------------------------------------------------
\subsubsection{Per-Query Identity and Bound}
\label{app:proof-eviction-identity}

For query $q$, softmax attention is normalized aggregation
readout with weights $w_i=e^{s_i(q)}$.
Appendix~\ref{app:scope-setup} gives Eq.~\ref{eq:intro-identity}.
Since $\bar v_E(q)$ and $o_M(q)$ are convex combinations of values
with $\|v_i\|_2\leq B_v$,
$\|\Delta o(q)\|_2
=\pi_E(q)\|\bar v_E(q)-o_M(q)\|_2
\leq 2B_v\pi_E(q).$

% ----------------------------------------------------------------------

% ----------------------------------------------------------------------
\subsubsection{Expected Eviction Certificate}
\label{app:proof-certificate}

Condition on the cache and eviction-time allocations.
The case $E=\varnothing$ is trivial;
otherwise Eq.~\ref{eq:intro-identity} and $\|v_i\|_2\le B_v$ give
$\mathbb E_q\|\Delta o(q)\|_2\le2B_v\bar p(E).$
Since $\pi_\theta^{\mathcal C}(E)=\mu_E$, total variation gives
\[
\bar p(E)\le
\min\!\left\{1,\,\mu_E+
\tfrac12\|\bar p-\pi_\theta^{\mathcal C}\|_1\right\}.
\]
The triangle inequality along
$\bar p\to\pi_u\to\pi_T^{\mathcal C}\to\pi_\theta^{\mathcal C}$
proves Eq.~\ref{eq:eviction-certificate}.
Utility mismatch includes query-aggregation and head effects;
it need not be small.

\textbf{Coverage shift.}
For $\beta\ge0$ and $\tau_{\mathrm{ent}}>0$, the finite teacher satisfies
$\pi_T^{\mathcal C}=\sigma((u+\beta d)/\tau_{\mathrm{ent}})$,
with $d=d(w^{(R-1)})$; convergence is unnecessary.
Set $x=u/\tau_{\mathrm{ent}}$, $h=\beta d/\tau_{\mathrm{ent}}$,
and $s(t)=\sigma(x+th)$.
Then $s_i'(t)=s_i(t)(h_i-\mathbb E_{s(t)}h)$ and
$\|s'(t)\|_1
=\mathbb E_{s(t)}|h-\mathbb E_{s(t)}h|
\le\sqrt{\operatorname{Var}_{s(t)}(h)}
\le\tfrac12\operatorname{osc}(h).$
The last inequality follows by centering at the range midpoint.
Integrating over $t\in[0,1]$ and using the simplex diameter yields
\begin{equation}
\delta_{\mathrm{cov}}
\le\min\!\left\{2,\,
\frac{\beta\operatorname{osc}(d)}{2\tau_{\mathrm{ent}}}\right\}.
\label{eq:app-coverage-shift}
\end{equation}
Pinsker's inequality further gives
$\|\pi_T^{\mathcal C}-\pi_\theta^{\mathcal C}\|_1
\le\sqrt{2D_{\mathrm{KL}}(
\pi_T^{\mathcal C}\|\pi_\theta^{\mathcal C})}$.

\textbf{Diagnostic interpretation.}
Define
$\mathcal D_E=\frac{\mathbb E_q\|\Delta o(q)\|_2}{\bar p(E)}$;
set $\mathcal D_E=0$ if $\bar p(E)=0$. Then
\begin{equation}
\mathbb E_q\|\Delta o(q)\|_2=\mathcal D_E\bar p(E),
\qquad 0\le\mathcal D_E\le2B_v.
\label{eq:app-directional-certificate}
\end{equation}
This measured direction factor is not used in the bound.
For any finite query collection, the same proof holds with empirical
averages, without independence assumptions; this does not establish
generalization to unobserved queries.

\textbf{Post-compensation geometry.}
Let $c(q)=g_\chi(q)m(q)$ denote the learned correction.
For nonzero residual and correction, define their alignment
and relative magnitude as
\begin{equation}
\rho(q)=
\frac{\langle\Delta o(q),c(q)\rangle}
{\|\Delta o(q)\|_2\|c(q)\|_2},
\qquad
r(q)=\frac{\|c(q)\|_2}{\|\Delta o(q)\|_2}.
\label{eq:diagnostic-correction-components}
\end{equation}
Expanding the squared residual gives the relative
squared-error reduction:
\begin{equation}
G_{\mathrm{comp}}(q)
=
1-
\frac{\|\Delta o(q)-c(q)\|_2^2}
{\|\Delta o(q)\|_2^2}
=
2r(q)\rho(q)-r(q)^2.
\label{eq:compensation-gain}
\end{equation}
Compensation reduces error iff $2r(q)\rho(q)>r(q)^2$.

\subsubsection{Decode-Time, Repeated-Eviction, and Multi-Head Extensions}
\label{app:decode-certificate}
\begin{corollary}[Current-query residual evaluation]
\label{cor:certified-eviction}
Given the current query's full-cache logits and values,
Eq.~\ref{eq:intro-identity} exactly determines $\pi_E$ and
$\Delta o$. A policy can enforce an error threshold
$\varepsilon$ or rank candidate eviction sets without additional
learned parameters, though evaluation still incurs computation.

This exact evaluation covers only queries present at eviction.
Theorem~\ref{thm:certificate} bounds population expected error; the impossibility result of
\citet{xie2026errorcertificates} precludes a consistent deterministic
estimate of a scheme's own unseen-query error from retained
information alone.
\end{corollary}

\begin{corollary}[Fixed-query accumulation under successive eviction]
\label{cor:additive}
Fix a query $q$ and key--value states, and let $\mathcal C_0\supseteq\cdots\supseteq\mathcal C_T$ be nonempty nested subsets. If $\|o(q;\mathcal C_{t-1})-o(q;\mathcal C_t)\|_2\le c_t(q)$ for each $t$, then
\(
\|o(q;\mathcal C_0)-o(q;\mathcal C_T)\|_2\le\sum_{t=1}^T c_t(q).
\)
\end{corollary}

\begin{remark}[Multi-head extension]
Apply Eq.~\ref{eq:intro-identity} independently to each head. After output projection,
$\|\Delta o_{\mathrm{out}}\|
\leq
\sum\nolimits_h
\|W_O^h\|\,
2B_v^h\pi_E^h.$
Because attention mass is head-specific, retaining per-head terms avoids collapsing distinct certificates into a single scalar. Summing the head-wise contributions therefore extends the bound directly through the output projection.
\end{remark}

% ======================================================================
\subsection{Allocation, Retention Fidelity, and Two-Channel Coupling}
\label{app:allocation-retention-coupling}

% ----------------------------------------------------------------------
\subsubsection{Allocation Decomposition}
\label{app:proof-allocation-decomposition}

Let $\pi\in\Delta_N$ and restrict it to a retained set $S\subseteq\mathcal C$ via
$
\pi_S(i)
=
\pi_i\mathbb I[i\in S],
$
with $E=\mathcal C\setminus S$. Nonnegativity gives
$
\|\pi-\pi_S\|_1
=
\sum_{i\notin S}\pi_i
=
1-\sum_{i\in S}\pi_i.
$
For $|S|=B$, minimizing residual allocation mass is therefore equivalent to maximizing retained $\pi$-mass:
\begin{equation}
\operatorname{TopB}(\pi)
\in\arg\min\nolimits_{|S|=B}\|\pi-\pi_S\|_1
=\arg\max\nolimits_{|S|=B}\sum\nolimits_{i\in S}\pi_i.
\label{eq:app-topb-projection}
\end{equation}
Top-$B$ is consequently the exact $B$-support projection of the allocation under residual $\ell_1$ loss.

Let
$
S_B=\operatorname{TopB}(\pi),
E_B=\mathcal C\setminus S_B,
\mu_B=\pi(E_B).
$
For $\mu_B>0$, define
$
\omega_i
=
\frac{\pi_i}{\mu_B},
i\in E_B.
$
Then
$
\omega_i\geq0,
\sum_{i\in E_B}\omega_i=1,
\pi_{E_B}=\mu_B\omega.
$
For any normalized representation
$\pi_{E_B}=\mu\widetilde\omega$, summing over $E_B$ gives
$\mu=\mu_B$, hence $\widetilde\omega=\omega$. For the
linear write $\mathcal A(r)=\sum_i r_iA_i$,
$
\mathcal A(\pi_{E_B})
=
\mu_B\sum\nolimits_{i\in E_B}\omega_iA_i.
$
Thus mass and conditional composition uniquely specify this residual allocation.

Assume $\|A_i\|\le C$. For any retained set
$S$ with $E=\mathcal C\setminus S$,
$
\left\|\sum\nolimits_{i\in E}\pi_iA_i\right\|
\le
\sum\nolimits_{i\in E}\pi_i\|A_i\|
\le
C\pi(E).
$
The bound is tight by choosing $A_i=Ca$ for a unit vector $a$, yielding
$
\sup\nolimits_{\|A_i\|\le C}
\left\|\sum\nolimits_{i\in E}\pi_iA_i\right\|
=
C\pi(E),
$
and Eq.~\ref{eq:app-topb-projection} implies
\begin{equation}
\operatorname{TopB}(\pi)
\in
\arg\min\nolimits_{|S|=B}
\sup\nolimits_{\|A_i\|\leq C}
\left\|
\sum\nolimits_{i\notin S}\pi_iA_i
\right\|.
\end{equation}

For CORE, choosing
$A_i=\phi_\varphi(k_i)v_i^\top$ or
$A_i=\phi_\varphi(k_i)\odot\phi_\varphi(k_i)$
gives the writes $\mu_E\bar M_E$ and $\mu_E\bar b_E$,
respectively. Thus retention and both memory writes use one
allocation; the memory map and gate learn the correction through
Eq.~\ref{eq:memory-loss}.
The minimax result concerns bounded linear writes, not realized
post-compensation error.
With protected entries fixed in $S$, the argument applies to
eligible candidates under the remaining budget.

% ----------------------------------------------------------------------
\subsubsection{Joint Retention and Excluded-Mass Fidelity}
\label{app:joint-fidelity}

Let $p=\pi_T^{\mathcal C}$ and $q=\pi_\theta^{\mathcal C}$
denote the teacher and student allocations.
Under a fixed tie-breaking rule, let
$S_p=\operatorname{TopB}(p)$ and
$S_q=\operatorname{TopB}(q)$, with complements
$E_p=\mathcal C\setminus S_p$ and
$E_q=\mathcal C\setminus S_q$.
Define the retention regret and excluded-mass discrepancy as
\begin{equation}
\operatorname{Regret}_B=p(S_p)-p(S_q),
\qquad
\operatorname{MassErr}_B=|p(E_p)-q(E_q)|.
\label{eq:student-fidelity}
\end{equation}

\paragraph{Retention fidelity.}
Since $S_p$ maximizes the retained $p$-mass and $S_q$ maximizes the
retained $q$-mass,
\[
\begin{aligned}
p(S_p)-p(S_q)
&=
\bigl[p(S_p)-q(S_p)\bigr]
+
\bigl[q(S_p)-q(S_q)\bigr]
+
\bigl[q(S_q)-p(S_q)\bigr]
\\
&\leq
\left|p(S_p)-q(S_p)\right|
+
\left|p(S_q)-q(S_q)\right|
\\
&\leq
\|p-q\|_1.
\end{aligned}
\]
Pinsker's inequality therefore gives
$0\leq p(S_p)-p(S_q)\leq
\sqrt{2D_{\mathrm{KL}}(p\|q)}$.

Let $p_{(1)}\geq\cdots\geq p_{(N)}$ denote the coordinates of $p$
in nonincreasing order, and define the cutoff margin
$\gamma_B=p_{(B)}-p_{(B+1)}$.
If $k_B=|S_p\setminus S_q|=|S_q\setminus S_p|$, then each exchanged
retained position loses at least $\gamma_B$ units of teacher mass.
Consequently,
\begin{equation}
k_B\gamma_B
\leq p(S_p)-p(S_q)
\leq \sqrt{2D_{\mathrm{KL}}(p\|q)}.
\label{eq:th-retention-exchange}
\end{equation}
It follows that
\begin{equation}
\sqrt{2D_{\mathrm{KL}}(p\|q)}
<
\gamma_B
\quad\Longrightarrow\quad
S_p=S_q.
\label{eq:th-exact-topb}
\end{equation}

The student recovers the teacher's Top-$B$ set when
the cutoff margin exceeds the distributional error bound. With four protected sinks and $4<B<N$, retention selects
the remaining $B-4$ positions among non-sink candidates.
The same bounds apply, with $\gamma_B$ taken at this cutoff
in the original teacher probabilities. Let $\mathcal U$ denote the non-sink candidates.
The KL term in Eq.~\ref{eq:indexer-loss} compares
their conditional allocations:
$D_{\mathcal U}
=
D_{\mathrm{KL}}\!\left(
\frac{p|_{\mathcal U}}{p(\mathcal U)}
\,\middle\|\,
\frac{q|_{\mathcal U}}{q(\mathcal U)}
\right).$
For the two protected retained sets, sink contributions
cancel in the regret. Applying the same argument
on $\mathcal U$ yields
\[
\operatorname{Regret}_B
\leq p(\mathcal U)\sqrt{2D_{\mathcal U}}.
\]

\paragraph{Boundary-focused supervision.}
Let $\mathcal I_{\mathrm{bd}}$ contain pairs
$
(i,j)
\in
P_{\mathrm{bd}}\times N_{\mathrm{bd}}
$
with
$
\widetilde s_i\leq\widetilde s_j.
$
Every such inversion contributes at least $m_{\mathrm{bd}}$ to the
boundary loss, so
\[
\frac{
|\mathcal I_{\mathrm{bd}}|
}{
|P_{\mathrm{bd}}||N_{\mathrm{bd}}|
}
\leq
\frac{
\mathcal L_{\mathrm{bd}}
}{
m_{\mathrm{bd}}
}.
\]
Thus $\mathcal L_{\mathrm{bd}}$ directly controls local inversion density around the deployed cutoff, although this local condition alone does not guarantee complete Top-$B$ agreement.

\paragraph{Excluded-mass fidelity.}
Let $\mathcal F_B$ be the common family of admissible retained sets,
and define
$F_B(r)=\max_{S\in\mathcal F_B}r(S)$.
Since $\mu_r=r(E_r)=1-F_B(r)$ for $r\in\{p,q\}$,
\begin{equation}
|\mu_p-\mu_q|
=|F_B(p)-F_B(q)|
\leq \max_{S\in\mathcal F_B}|p(S)-q(S)|
\leq \operatorname{TV}(p,q).
\label{eq:th-mass-fidelity-tv}
\end{equation}
Pinsker's inequality yields
\begin{equation}
|\mu_p-\mu_q|
\leq \sqrt{\tfrac12D_{\mathrm{KL}}(p\|q)},
\label{eq:th-mass-fidelity-kl}
\end{equation}
even when $S_p\neq S_q$.
Full-allocation KL bounds retention regret and excluded-mass
discrepancy; conditional training KL yields the protected-retention
bound, while full-mass fidelity is evaluated in Appendix~\ref{app:retention-fidelity}.
For vector-valued writes, changes in the evicted set contribute
an additional term that vanishes when $E_p=E_q$
(Appendix~\ref{app:proof-memory-stability}).

% ----------------------------------------------------------------------

\subsection{Coverage Interaction and Set Awareness}
\label{app:proof-cache-conditioned-coverage}

Set
$
x_i=z_i^c,
\Sigma(w)
=
I+\alpha\sum_jw_jx_jx_j^\top,
F(w)=\log\det\Sigma(w).
$
For this analysis, we extend $F$ to $w\in\mathbb R_+^N$
with fixed $x_i$ and nonnegative coordinate increments. Jacobi's identity together with the inverse derivative rule yields
\begin{equation}
d_i(w)
=
\frac{\partial F(w)}{\partial w_i}
=
\alpha x_i^\top\Sigma(w)^{-1}x_i,
\qquad
\frac{\partial d_i}{\partial w_j}
=
-\alpha^2
\big(
x_i^\top\Sigma(w)^{-1}x_j
\big)^2
\leq0.
\label{eq:th-cross-curvature}
\end{equation}
The nonpositive cross derivative establishes coordinate-wise diminishing returns and identifies squared inverse-covariance overlap as the exact local interaction between candidate directions.

\paragraph{Pairwise finite gap.}
With
$
\Delta_i(a;w)
=
F(w+ae_i)-F(w)
$
and
$
\Delta_{ij}(a,b;w)
=
F(w+ae_i+be_j)-F(w),
$
the fundamental theorem of calculus gives
\begin{align}
G_{ij}(a,b;w)
&=
\Delta_i(a;w)
+
\Delta_j(b;w)
-
\Delta_{ij}(a,b;w)
\nonumber\\
&=
-\int_0^a\!\!\int_0^b
\frac{
\partial^2F(w+se_i+te_j)
}{
\partial w_i\partial w_j
}
\,dt\,ds
\nonumber\\
&=
\alpha^2
\int_0^a\!\!\int_0^b
\left[
x_i^\top
\Sigma(w+se_i+te_j)^{-1}
x_j
\right]^2
dt\,ds
\geq0.
\label{eq:th-finite-nonadditivity}
\end{align}

If
$
\left|
x_i^\top
\Sigma(w+se_i+te_j)^{-1}
x_j
\right|
\leq
\varepsilon
$
throughout the integration region, then
$
0
\leq
G_{ij}(a,b;w)
\leq
\alpha^2ab\varepsilon^2.
$
By continuity,
\[
\lim_{a,b\to0}
\frac{G_{ij}(a,b;w)}{ab}
=
\alpha^2
\big(
x_i^\top\Sigma(w)^{-1}x_j
\big)^2.
\]
Cache-whitened orthogonality therefore removes the local interaction, whereas larger inverse-covariance overlap produces a correspondingly larger departure from additive scalarized gains.

\paragraph{Finite-set extension.}
For a finite set $S$ and increments $\{a_i\}_{i\in S}$, let
\[
G_S
=
\sum\nolimits_{i\in S}\Delta_i(a_i;w)
-
\left[
F\!\left(
w+\sum\nolimits_{i\in S}a_ie_i
\right)
-
F(w)
\right].
\]
Order the elements of $S$ arbitrarily and telescope the joint gain. Since every cross derivative in Eq.~\ref{eq:th-cross-curvature} is nonpositive, each newly added coordinate has marginal gain no larger than at the common base state. Therefore
$
G_S\geq0.
$
If
$
\left|
x_i^\top\Sigma(u)^{-1}x_j
\right|
\leq
\varepsilon_{ij}
$
throughout the coordinate box between
$w$ and $w+\sum_{i\in S}a_ie_i$, repeated integration gives
$
G_S
\leq
\alpha^2
\sum\nolimits_{i<j}
a_ia_j\varepsilon_{ij}^2.
$
Uniformly weak inverse-covariance interactions therefore imply
near-additive gains from a shared cache state. This conclusion concerns factorization from a shared cache state; it does not make the student ranking an exact optimizer of the original fixed-budget log-determinant objective.

\textbf{Relation to the prototype features.}
For a reliable candidate,
$
\widehat z_i^c
=
\frac{z_i^c}{\|z_i^c\|_2}.
$
Since $\widehat z_i^c$ and each prototype have unit norm,
$
d_i^{\mathrm{proto}}
=
1-(\widehat z_i^c)^\top p_{m_1(i)}
=
\frac12
\left\|
\widehat z_i^c-p_{m_1(i)}
\right\|_2^2.
$
If every prototype similarity changes by at most $\delta$, then
$
\Delta_i^{\mathrm{proto}}>2\delta
$
implies that $m_1(i)$ is unchanged.

Prototype distance represents directional residual, prototype margin represents assignment robustness, and occupancy represents local directional crowding. These finite-dimensional cues provide a compact interface through which the indexer learns cache-level geometry distilled from the teacher.

% ======================================================================
\subsection{Mass-Preserving and Recurrent Memory}
\label{app:mass-preserving-recurrent-memory}

% ----------------------------------------------------------------------
\subsubsection{Proof of Proposition~\ref{prop:memory-stability}}
\label{app:proof-memory-stability}

For the set-aware stability argument, let
$
E_p
=
\mathcal C\setminus\operatorname{TopB}(p),
E_q
=
\mathcal C\setminus\operatorname{TopB}(q),
$
with symmetric difference
$
\Delta_E=E_p\triangle E_q.
$
Write
$r(E)=\sum\nolimits_{i\in E}r_i,
R(r,E)=\sum\nolimits_{i\in E}r_iA_i,$
and define
\begin{equation}
\eta(p,q)
=
\|p-q\|_1
+
\min\{p(\Delta_E),q(\Delta_E)\}.
\label{eq:app-set-aware-eta}
\end{equation}

\textbf{Mass-scaled write.}
We separate probability perturbation from set-membership perturbation:
\begin{align*}
R(p,E_p)-R(q,E_q)
&=
[R(p,E_p)-R(q,E_p)]
+
[R(q,E_p)-R(q,E_q)].
\end{align*}
When $\|A_i\|\leq C$, the first term is at most $C\|p-q\|_1$, while the second differs only on $\Delta_E$ and is at most $Cq(\Delta_E)$. Decomposing symmetrically through $R(p,E_q)$ gives the corresponding $p(\Delta_E)$ bound; taking the tighter one yields
\begin{equation}
\|R(p,E_p)-R(q,E_q)\|
\leq
C\eta(p,q).
\label{eq:app-mass-scaled-stability}
\end{equation}

\textbf{Conditionally normalized write.}
Let
$
\mu_p=p(E_p),
\mu_q=q(E_q),
\mu_{\min}=\min\{\mu_p,\mu_q\}>0,
$
and define
$\bar A(r,E)
=
\frac{R(r,E)}{r(E)}.$
Eq.~\ref{eq:th-mass-fidelity-tv} gives
$|\mu_p-\mu_q|
\leq
\eta(p,q).$
Decompose
\begin{align*}
\bar A(p,E_p)-\bar A(q,E_q)
&=
\frac{
R(p,E_p)-R(q,E_q)
}{
\mu_p
}+
\left(
\frac{1}{\mu_p}
-
\frac{1}{\mu_q}
\right)
R(q,E_q).
\end{align*}
Using
$
\|R(q,E_q)\|
\leq
C\mu_q
$
yields
\begin{equation}
\|\bar A(p,E_p)-\bar A(q,E_q)\|
\leq
\frac{2C}{\mu_{\min}}
\eta(p,q).
\label{eq:app-conditional-sensitivity}
\end{equation}
Conditional renormalization introduces exactly one new factor, $1/\mu_{\min}$; the allocation-preserving write avoids this inverse-mass sensitivity.

\textbf{KL specialization.}
Pinsker's inequality gives
$
\|p-q\|_1
\leq
\sqrt{2D_{\mathrm{KL}}(p\|q)},
$
hence
$
\|R(p,E_p)-R(q,E_q)\|
\leq
C
\left[
\sqrt{2D_{\mathrm{KL}}(p\|q)}
+
\min\{p(\Delta_E),q(\Delta_E)\}
\right].
$
If $E_p=E_q$, including any case satisfying Eq.~\ref{eq:th-exact-topb}, the membership term vanishes and write discrepancy is controlled only by KL.

\textbf{Inverse-mass dependence is attainable.}
Take three items with retained budget $B=1$, and let
$0<\mu<1/2$ and $0<\delta<\mu/2$.
Set
$q=(\mu/2,\mu/2,1-\mu)$ and
$p=(\mu/2+\delta,\mu/2-\delta,1-\mu)$.
Both allocations retain item $3$, so $E_p=E_q=E=\{1,2\}$ and
$\mu_{\min}=\mu$.
Choose $A_1=a$ and $A_2=-a$ with $\|a\|=C$.
Since $\|p-q\|_1=2\delta$ and both excluded masses equal $\mu$,
\[
\frac{\|\bar A(p,E)-\bar A(q,E)\|}{\|p-q\|_1}
=\frac{C}{\mu},
\qquad
\frac{\|R(p,E)-R(q,E)\|}{\|p-q\|_1}
=C.
\]
Thus inverse-mass dependence is unavoidable in general for
conditional normalization.

\subsubsection{Hierarchical Stability of Mass-Preserving Memory}
\label{app:recurrent-memory-characterization}

We characterize how the write-level stability of
Proposition~\ref{prop:memory-stability} propagates through CORE's recurrent
memory and query-adaptive readout. The resulting bounds separate three
sources of sensitivity: allocation perturbation, accumulated-mass
normalization, and readout normalization.

At compression step $s$, let
$
\pi_{s,i}
=
\pi_\theta^{\mathcal C_s}(i),
\mu_s
=
\sum_{i\in E_s}\pi_{s,i},
$
and define
$
A_{s,i}
=
\phi_\varphi(k_{s,i})v_{s,i}^{\top},
B_{s,i}
=
\phi_\varphi(k_{s,i})
\odot
\phi_\varphi(k_{s,i}).
$
Appendix~\ref{app:proof-allocation-decomposition} and
Eq.~\ref{eq:evicted-mass} give
\begin{equation}
\mu_s\bar M_{E_s}
=
\sum\nolimits_{i\in E_s}\pi_{s,i}A_{s,i},
\qquad
\mu_s\bar b_{E_s}
=
\sum\nolimits_{i\in E_s}\pi_{s,i}B_{s,i}.
\label{eq:app-step-write}
\end{equation}
Each event writes a linear image of the residual allocation,
not a unit-mass conditional summary.

\begin{corollary}[Hierarchical perturbation stability]
\label{cor:recurrent-stability}
Consider two allocation trajectories
$\{p_s\}_{s=1}^{t}$ and $\{q_s\}_{s=1}^{t}$ over the same
candidate memory atoms, with common initialization
$(Z_0^r,M_0^r,b_0^r)=(0,0,0)$ for $r\in\{p,q\}$ and
$\|A_{s,i}\|_F\le C_M$, $\|B_{s,i}\|_2\le C_b$.
Let
$
\eta_s=\eta(p_s,q_s),
\mathcal E_t
=
\sum_{s=1}^{t}\gamma_M^{\,t-s}\eta_s .
$
For $r\in\{p,q\}$, define
$
D_t^r=Z_t^r+\epsilon,
a_t^r=\frac{\gamma_MZ_{t-1}^r}{D_t^r},
$
and
$
D_{t,\min}=\min_r D_t^r,
\bar a_t=\max_r a_t^r,
L_0=0,
L_t=
\bar a_tL_{t-1}
+
\frac{4\mathcal E_t}{D_{t,\min}}.
$
Then
$
\|M_t^p-M_t^q\|_F
\le C_ML_t,
\|b_t^p-b_t^q\|_2
\le C_bL_t.
$
For a fixed query, let
$h=\phi_\varphi(q)$,
$d_r=(h\odot h)^\top b_t^r+\epsilon$,
and
$d_{\min}=\min_r d_r$.
Then
\begin{equation}
\|m_t^p(q)-m_t^q(q)\|_2
\le
C_ML_t\|h\|_2
\left(
\frac{1}{\sqrt{d_{\min}}}
+
\frac{
C_b\|h\odot h\|_2
}{
2d_{\min}^{3/2}
}
\right).
\label{eq:hierarchical-memory-stability}
\end{equation}
Since $g_\chi(q)\in[0,1]$ depends only on the fixed query,
the same upper bound applies to
$\|g_\chi(q)m_t^p(q)-g_\chi(q)m_t^q(q)\|_2$. With fixed retained support, the same bound therefore holds for the
compensated output difference.
\end{corollary}

\begin{proof}
Proposition~\ref{prop:memory-stability}, applied separately to
$A_{s,i}$ and $B_{s,i}$, gives
$
\|R_{M,s}^p-R_{M,s}^q\|_F\le C_M\eta_s,
\|R_{b,s}^p-R_{b,s}^q\|_2\le C_b\eta_s,
$
while
$|\mu_s^p-\mu_s^q|\le\eta_s$.
Hence
$
|Z_t^p-Z_t^q|
\le
\mathcal E_t .
$
Writing the recurrence as
\[
M_t^r
=
a_t^rM_{t-1}^r
+
\frac{R_{M,t}^r}{D_t^r}
\]
gives
$\|M_t^r\|_F\le C_M$,
\[
|a_t^p-a_t^q|
\le
\frac{2\mathcal E_t}{D_{t,\min}},
\]
and
\[
\left\|
\frac{R_{M,t}^p}{D_t^p}
-
\frac{R_{M,t}^q}{D_t^q}
\right\|_F
\le
\frac{2C_M\mathcal E_t}{D_{t,\min}} .
\]
Therefore
\[
\|M_t^p-M_t^q\|_F
\le
\bar a_t\|M_{t-1}^p-M_{t-1}^q\|_F
+
\frac{4C_M\mathcal E_t}{D_{t,\min}},
\]
and the same argument applies to $b_t$.
For the readout, add and subtract
$h^\top M_t^q/\sqrt{d_p}$ and use
\[
|x^{-1/2}-y^{-1/2}|
\le
\frac{|x-y|}{2d_{\min}^{3/2}},
\qquad
|d_p-d_q|
\le
\|h\odot h\|_2
\|b_t^p-b_t^q\|_2 .
\]
Substitution gives Eq.~\ref{eq:hierarchical-memory-stability}.
The gate inequality follows from $0\le g_\chi(q)\le1$.
\end{proof}

Together, Proposition~\ref{prop:memory-stability} and
Corollary~\ref{cor:recurrent-stability} give a hierarchical stability
decomposition: linear writes avoid inverse-mass amplification, recurrent
normalization is controlled by $D_{t,\min}$, and query-adaptive
normalization by $d_{\min}$; the scalar gate is non-expansive.

\section{Theory-Aligned Empirical Diagnostics for Sec. \ref{sec:mechanistic-analysis}}
\label{app:diagnostic-protocol}

We use 104 RULER-16K sequences (eight per task) from
Llama-3.1-8B and four layers spanning model depth.
All diagnostic evaluations freeze learned parameters and share
reference QKV states, query visibility, protected positions,
and tie-breaking.
The budget is $B=\lfloor(1-r)N_0\rfloor$, including four sink tokens.
Retention and allocation diagnostics use
$r\in\{0.25,0.50,0.75,0.90\}$ and up to 32 subsequent queries.
Attention metrics are computed per head and then averaged over
heads and layers within sequences, then equally over tasks; allocation metrics are
computed per sequence--layer state.
Comparisons use common valid observations and exclude
undefined ratios.

\subsection{Retention and Allocation Fidelity}
\label{app:retention-fidelity}

\textbf{Coverage and Retention Error.}
With memory correction disabled, we compare CORE with the
$\beta=0$ selector on shared reference states.
For each sequence--layer--head state $g$ with evaluation query
set $\mathcal Q_g$, compute the empirical counterparts
$\bar p_g(E)$ and $\widehat{\mathcal D}_{E,g}$ of
Eq.~\ref{eq:app-directional-certificate} using uniform averages
over $\mathcal Q_g$.
Normalize the directional factor as
\begin{equation}
\widetilde{\mathcal D}_{E,g}
=
\frac{\widehat{\mathcal D}_{E,g}}{2B_{v,g}},
\qquad
B_{v,g}=\max_i\|v_i\|_2.
\label{eq:retention-error}
\end{equation}
Fig.~\ref{fig:mechanistic-analysis}(a) reports both factors and
relative reductions over $\beta=0$ in their mean product,
$\langle\bar p_g(E)\widetilde{\mathcal D}_{E,g}\rangle$.
Normalized errors are computed per state before averaging.

\begin{wraptable}{r}{0.5\textwidth}
\centering
\footnotesize
\setlength{\tabcolsep}{3pt}
\renewcommand{\arraystretch}{1.0}
\caption{Empirical normalized eviction-certificate audit.
All columns use common state weights.
$M$ and Bound agree at the displayed precision.
Tightness $=\mathrm{RetErr}/\mathrm{Bound}$.}
\label{tab:certificate-audit}
\resizebox{\linewidth}{!}{%
\begin{tabular}{c|ccccc}
\toprule
Comp.
& $\mu_E$
& $\bar p(E)$
& $M$, Bound
& RetErr
& Tight. \\
\midrule
25\%
& 0.120 & 0.3251 & 0.370 & 0.2471 & 0.668 \\
50\%
& 0.300 & 0.5427 & 0.550 & 0.4354 & 0.792 \\
75\%
& 0.520 & 0.7436 & 0.770 & 0.5289 & 0.687 \\
90\%
& 0.680 & 0.8718 & 0.930 & 0.5392 & 0.580 \\
\bottomrule
\end{tabular}%
}
\vspace{-0.2in}
\end{wraptable}

\textbf{Certificate audit.}
For state $g$, let $\kappa_g$ be half the discrepancy sum in
Eq.~\ref{eq:eviction-certificate}, and let
$\mu_{g,r}=\pi_{\theta,g}(E_{g,r})$.
Using the same states and weights as
Fig.~\ref{fig:mechanistic-analysis}(a),
Table~\ref{tab:certificate-audit} reports
$M=\langle\mu_{g,r}+\kappa_g\rangle$,
$\mathrm{Bound}
=\langle\min\{1,\mu_{g,r}+\kappa_g\}\rangle$,
and
$\mathrm{RetErr}
=\langle\bar p_g(E)\widetilde{\mathcal D}_{E,g}\rangle$.
The mean distillation, coverage-shift, and utility-mismatch
contributions, including the factor $1/2$, are
0.014, 0.078, and 0.158, summing to 0.250.
The gap between attention mass and RetErr reflects the conservative
$2B_v$ directional envelope.

\begin{wraptable}{r}{0.5\textwidth}
\centering
\footnotesize
\vspace{-0.3in}
\setlength{\tabcolsep}{3pt}
\renewcommand{\arraystretch}{0.72}
\caption{Fixed-budget utility--coverage set quality.}
\label{tab:teacher-greedy}
\begin{tabular}{llrrr}
\toprule
Compression & Selector & $U_B$ & $C_B$ & $J_B$ \\
\midrule
75\%
& Utility-only
& 1.3689 & 0.9883 & 2.3572 \\
& Student Top-$B$
& 1.1548 & 1.3262 & 2.4810 \\
& Teacher Top-$B$
& 1.1503 & 1.3334 & 2.4837 \\
& Greedy
& 1.1140 & 1.3911 & \textbf{2.5051} \\
\midrule
90\%
& Utility-only
& 1.4352 & 0.9351 & 2.3703 \\
& Student Top-$B$
& 1.1522 & 1.3616 & 2.5138 \\
& Teacher Top-$B$
& 1.1468 & 1.3764 & 2.5232 \\
& Greedy
& 1.1229 & 1.4264 & \textbf{2.5493} \\
\bottomrule
\end{tabular}
\vspace{-0.2in}
\end{wraptable}

\textbf{Fixed-budget set quality.} We compare utility-only, Student Top-$B$, Teacher Top-$B$, and
a discrete greedy reference under shared selection constraints
and the objective in Eq.~\ref{eq:app-fixed-budget-objective}.
Here $U_B$ and $C_B$ are its utility and $\beta$-weighted
log-determinant terms. Define
$\operatorname{Gap}_{\mathrm{greedy}}(S)
=[J_B(S_{\mathrm{greedy}})-J_B(S)]/|J_B(S_{\mathrm{greedy}})|$.
At 75\%/90\% compression, gaps are 0.96\%/1.39\% for Student,
0.85\%/1.02\% for Teacher, and 5.90\%/7.02\% for utility-only
(Table~\ref{tab:teacher-greedy}).

\textbf{Allocation fidelity.}
We compare CORE, the no-coverage-cues student, and the KL-only
student against the same teacher under shared selection constraints.
We report $\operatorname{Regret}_B$ and $\operatorname{MassErr}_B$
from Eq.~\ref{eq:student-fidelity}, measuring lost teacher-assigned
retained mass and excluded-mass discrepancy, respectively.
Both errors are lower for CORE at every tested budget
(Fig.~\ref{fig:mechanistic-analysis}(b)).

\subsection{Memory Reconstruction and Stability}
\label{app:compensation-diagnostics}

\textbf{Write stability.}
We test the sensitivity characterized by
Proposition~\ref{prop:memory-stability}. For each deployed allocation $p=\pi_\theta^{\mathcal C}$, fix keys,
values, and the memory map and draw four perturbations per level:
\begin{equation}
q_\delta=\operatorname{softmax}(\log p+\delta\xi),
\qquad
\xi\sim\mathcal N(0,I),
\qquad
\delta\in\{0.01,0.02,0.05\}.
\label{eq:diagnostic-allocation-perturbation}
\end{equation}
Eviction sets are recomputed under shared selection constraints.
We instantiate the notation of
Appendix~\ref{app:proof-memory-stability} with
$q=q_\delta$, $A_i=\phi_\varphi(k_i)v_i^\top$, and
$C=\max_i\|A_i\|_F$.
The normalized sensitivities:
\begin{equation}
S_{\mathrm{mass}}
=
\frac{\|R(p,E_p)-R(q,E_q)\|_F}{C\eta(p,q)},
\qquad
S_{\mathrm{cond}}
=
\frac{\|\bar A(p,E_p)-\bar A(q,E_q)\|_F}{C\eta(p,q)}.
\label{eq:diagnostic-write-sensitivity}
\end{equation}
Valid records with $C\eta(p,q)>0$ and $\mu_{\min}>0$
are grouped using bin edges $\{0.05,0.10,0.20,0.40,0.80\}$;
the plot reports medians and 90th percentiles at geometric
bin midpoints.
Fig.~\ref{fig:mechanistic-analysis}(c) shows that
mass-preserving sensitivity stays approximately constant,
whereas conditional-only sensitivity increases as excluded
mass decreases, consistent with
Proposition~\ref{prop:memory-stability}.

\textbf{Memory reconstruction.}
Following the post-compensation geometry in
Appendix~\ref{app:proof-certificate}, we assess reconstruction
via compensation gain, alignment, and relative magnitude.
We compare mass-preserving accumulation with equal-event writing
(unit event mass, unchanged conditional weights) and uniform-token
writing (unchanged event mass, $\omega_i=1/|E|$)
using Eq.~\ref{eq:evicted-mass}.
Memory maps and gates are trained separately with matched
initialization, data, and  optimization budgets.

Replay variants share eviction events, memory capacity, decay,
and stabilizers under the common frozen-module protocol above.
Memory resets per sequence. The final 512 prompt tokens form four
128-token blocks, each preceded by recompression to
$B=\lfloor0.25N_0\rfloor$ ($N_0$: initial prefix length)
and evaluated on its first 16 queries.
At each write event, for each sequence--layer group $g$,
let $\Delta o(q)=o_{\mathrm{full}}(q)
-o_{\mathrm{attn}}(q;\mathcal S_{\mathrm{keep}})$ and
$c(q)=g_\chi(q)m(q)$ denote the full pre-$W_O$ residual
and correction in $\mathbb R^D$.
We concatenate these vectors over the evaluated queries into
$\Delta_g$ and $c_g$, respectively.
Applying Eqs.~\ref{eq:diagnostic-correction-components}
and~\ref{eq:compensation-gain} to the concatenated vectors gives
groupwise alignment $\rho_g$, relative magnitude $r_g$,
and compensation gain $G_{\mathrm{comp},g}$.
The gain therefore equals the relative reduction in summed
squared error across the evaluated queries and pre-$W_O$
coordinates within each sequence--layer group.
Alignment is evaluated on common groups with nonzero
residual and correction norms.
All three metrics are computed per group before averaging.

\begin{figure}[t]
\centering
\includegraphics[width=\linewidth]
{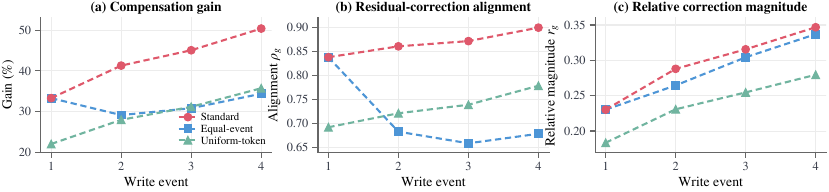}
\vspace{-0.3in}
\caption{
Repeated-write compensation at 75\% initial compression:
(a) compensation gain,
(b) residual--correction alignment $\rho_g$, and
(c) relative correction magnitude $r_g$.
}
\vspace{-0.2in}
\label{fig:memory-correction-geometry}
\end{figure}

\begin{wraptable}{r}{0.45\columnwidth}
\vspace{-0.1in}
\centering
\footnotesize
\setlength{\tabcolsep}{4pt}
\renewcommand{\arraystretch}{0.7}
\caption{Fourth-event compensation gain at $\delta=0.05$. Gains are percentages; $\Delta G=G^{(0)}-G^{(\delta)}$ is the percentage-point decrease.}
\label{tab:synthetic-perturbed-readout}
\begin{tabular}{lccc}
\toprule
Writing rule & $G^{(0)}$ (\%) & $G^{(\delta)}$ (\%) & $\Delta G$ \\
\midrule
Mass-weighted & 50.38 & 49.72 & 0.66 \\
Equal-event & 34.38 & 33.94 & 0.44 \\
Uniform-token & 35.73 & 35.16 & 0.57 \\
\bottomrule
\end{tabular}
\vspace{-0.1in}
\end{wraptable}

Mass-weighted gain rises from 33.30\% to 50.38\% over four events,
ending 16.00 and 14.65 percentage points above equal-event and
uniform-token writing, respectively
(Fig.~\ref{fig:memory-correction-geometry}).
Compared with equal-event writing, it achieves higher later alignment
at a similar correction magnitude.
Separately trained mass-weighted and equal-event variants have similar
first-write gains; with identical maps and gates and $\epsilon=0$,
their single-event readouts coincide.
Event mass controls relative weighting across subsequent writes.

To isolate write-weight sensitivity, we apply
Eq.~\ref{eq:diagnostic-allocation-perturbation} at every write event,
fixing eviction sets and sharing noise draws across variants.
At $\delta=0.05$, fourth-event gains decrease by 0.44--0.66
percentage points; mass-weighted writing remains highest after
perturbation (Table~\ref{tab:synthetic-perturbed-readout}).

\providecommand{\coretodo}[1]{\textbf{[To confirm: #1]}}

\providecommand{\coretodo}[1]{\textbf{[To confirm: #1]}}

\section{Implementation and Evaluation Protocol}
\label{app:reporting-protocol}

\subsection{Implementation Details}
\label{app:implementation-details}

We evaluate CORE on Mistral-7B-v0.2,
Llama-3.1-8B-Instruct, and Qwen3-14B.
The backbone remains frozen throughout training; only the indexer
and compensation modules are optimized.

All evaluated backbones use GQA,
with $H_q>H_{kv}$.
For memory computation in
Eqs.~\ref{eq:memory-update}--\ref{eq:memory-readout},
let $g(h)$ map query head $h$ to its shared KV head.
We use query-head-aligned representations
\[
q=\operatorname{Concat}_{h=1}^{H_q} q^{(h)},\qquad
k_i=\operatorname{Concat}_{h=1}^{H_q} k_i^{(g(h))},\qquad
v_i=\operatorname{Concat}_{h=1}^{H_q} v_i^{(g(h))},
\]

\begin{wraptable}{r}{0.5\textwidth}
\vspace{-0.3in}
\centering
\scriptsize
\setlength{\tabcolsep}{4pt}
\renewcommand{\arraystretch}{0.8}
\caption{Implementation details.}
\label{tab:implementation-details}
\begin{tabular}{ll}
\toprule
\textbf{Component} & \textbf{Setting} \\
\midrule
Indexer
& $14\!\rightarrow\!128\!\rightarrow\!128\!\rightarrow\!1$, GELU \\
Projected key space
& $d_z=64$ \\
Coverage prototypes
& $M_p=64$, fixed \\
Memory map
& $\phi_\varphi:\mathbb R^D\!\rightarrow\!\mathbb R^{512}$, no bias \\
Memory state
& $M_t\!\in\!\mathbb R^{512\times D}$,
  $b_t\!\in\!\mathbb R^{512}$, $Z_t\!\in\!\mathbb R$ \\
Memory readout
& $m(q)\!\in\!\mathbb R^D$ (pre-$W_O$ space) \\
Gate
& $D\!\rightarrow\!64\!\rightarrow\!1$, GELU + sigmoid \\
Mistral/Llama memory map
& Dense \\
Qwen3-14B memory map
& Block diagonal, $4\times[128\times(D/4)]$ \\
\midrule
Coverage teacher
& $\alpha=1$, $\beta=128$, $\tau_{\rm ent}=8$, $R=8$ \\
TopMean
& $K=4$ \\
$f_6$
& Top 25\% of causally visible queries \\
Recent queries
& $|Q_{\rm recent}|=16$ \\
Coverage descriptors
& $\rho_{\rm sh}=0.1$, $\epsilon_{\rm geo}=10^{-6}$ \\
Boundary supervision
& $\lambda_{\rm bd}=4$, $\Delta_B=32$, $m_{\rm bd}=0.5$ \\
Allocation temperature
& $\tau_M=1.0$ \\
Memory decay
& $\gamma_M=0.9$ \\
Numerical stabilizer
& $\epsilon=10^{-8}$ \\
Protected tokens
& First four sink tokens \\
\bottomrule
\end{tabular}
\vspace{-0.3in}
\end{wraptable}

so that $q,k_i,v_i\in\mathbb R^D$ with
$D=H_qd_{\mathrm{head}}$.
Queries and keys use post-RoPE coordinates.
At inference, this GQA expansion is not materialized:
repeated KV blocks are stored once, and the fixed head repetition
is folded into the corresponding projections.

Each Transformer layer uses independent module parameters shared
across token positions.
For Qwen3-14B, the block-diagonal memory map preserves the
512-dimensional latent space while partitioning the input and latent
coordinates into four groups.

We use FlashAttention-2 for backbone attention and compute CORE's
relation features through tiled QK reductions, with custom Triton
kernels for supported GQA configurations.
The optimized implementation avoids materializing the full attention
matrix and reuses FlashAttention softmax log-normalizers when available;
otherwise, they are computed through online log-sum-exp accumulation.
Protected sink tokens count toward the retained-KV budget and are
excluded from ranking supervision and memory writes.
No additional recency bonus or question-token protection is used.

\subsection{Training Details}
\label{app:d.2}

\textbf{Data and optimization.}
We train on the full LongAlpaca-12k dataset.
With an effective batch size of four, training comprises
1,000 Stage~I indexer-distillation steps and 4,400 Stage~II
joint training steps, totaling 5,400 updates.
Optimization uses AdamW with
$(\beta_1,\beta_2)=(0.9,0.999)$, $\epsilon=10^{-8}$,
weight decay 0.01, and gradient-norm clipping at 1.0.
A continuous base learning-rate schedule spans both stages,
with 100 warmup steps, 3,100 constant-rate steps, and
2,200 cosine-decay steps.
The peak and final base learning rates are
$5\times10^{-4}$ and $7.5\times10^{-6}$, respectively.
Stage~II applies learning-rate multipliers of 0.1 to the
indexer and 1.0 to the memory modules.
The training seed is 42.

\textbf{Sequence construction and supervision.}
The maximum training length is 2,048 tokens, increased to 4,096 every eighth optimizer step.
Overlength contexts are cropped deterministically by step, preserving the initial sink tokens and question/answer suffix.
Training uses compression ratios $\{0.25,0.50,0.75,0.90\}$.
Prefill boundary supervision averages over all four cutoffs, while memory writing and single-budget decode supervision use one uniformly sampled ratio per step.
Additional decode supervision is applied every fourth step using up to 128 recent queries with unit loss weight.

Memory reconstruction uses up to 16 trailing response queries, or a held-out prompt suffix when responses are unavailable.
Only evicted tokens from the preceding prefix are written into memory.
Memory states are rebuilt independently for each example.
The reconstruction loss updates the memory map and gate without propagating gradients through hard selection or into indexer scores.

\subsection{Evaluation Protocol}
\label{app:d.3}
CORE and all baselines follow the same KVPress evaluation
pipeline, including benchmark inputs, prompt templates,
demonstrations, generation settings, and scoring procedures.
Each model uses its associated tokenizer and chat template.
Qwen3 evaluation sets \texttt{enable\_thinking=False}.
Unless otherwise specified, we use seed 42 and one greedy
completion per question, retaining KVPress's benchmark-specific
output limits and context-truncation rules.

\textbf{RULER and LongBench.}
RULER includes 13 subtasks with 500 examples per subtask
at each reported context length.
LongBench evaluation covers 14 English tasks and two code-completion tasks.
For both benchmarks, we report the unweighted mean of
the task scores.
The six-task LongBench radar plot is used only for visualization
and does not define the aggregate score.

\textbf{AIME25 and Math500.}
We evaluate all 30 and 500 examples from
\texttt{alessiodevoto/aime25} and
\texttt{alessiodevoto/math500}~\citep{ICLR2024_aca97732}, respectively.
We retain KVPress's prompts and maximum output
lengths of 32,000 and 4,096 tokens.
Accuracy for all examples uses the adopted KVPress
implementation's boxed-answer extraction and exact
string matching.

\subsection{Development Set and Hyperparameter Selection}
\label{app:validation-selection}

We used 320 held-out development examples covering LongBench's 16 English and code tasks, sampling 20 examples per task by input length.
Data were drawn from available source-development splits and independent samples for the remaining tasks, and screened for training/test overlap using source IDs and normalized text.
Hyperparameter selection used only Llama-3.1-8B-Instruct on LongBench at 75\% compression, maximizing the unweighted task average under greedy decoding.
The selected hyperparameters were held fixed across all backbones, benchmarks, and budgets.
Each method evaluated at most ten candidate configurations, including the initial setting.
Baseline candidates included their reported best settings and neighboring values, and actual tuning costs were recorded for each method.
The sensitivity analyses in Appendix~\ref{app:hyperparameter-sensitivity} assessed robustness after the hyperparameters were fixed.
Their results did not guide hyperparameter selection or revision.

\section{Additional Experimental Results}
\label{app:additional-experiments}

\subsection{Per-task Results on LongBench and RULER at 75\% Compression}
\label{app:detailed-results}
This section resolves that summary into task-level results at a fixed 75\% compression ratio, reporting every LongBench task and every RULER-16K subtask for Llama-3.1-8B under the retained-KV budget used by the corresponding aggregate comparison.

\begin{table*}[t]
\centering
\caption{Per-task LongBench results for Llama-3.1-8B at 75\% compression. Tasks follow the official LongBench taxonomy; Avg. is the unweighted mean over all 16 tasks.}
\label{tab:longbench-per-task}
\setlength{\tabcolsep}{2.8pt}
\resizebox{\textwidth}{!}{
\begin{tabular}{lccccccccccccccccc}
\toprule
& \multicolumn{3}{c}{\textbf{SD-QA}}
& \multicolumn{3}{c}{\textbf{MD-QA}}
& \multicolumn{3}{c}{\textbf{Summ.}}
& \multicolumn{3}{c}{\textbf{FSL}}
& \multicolumn{2}{c}{\textbf{Synth.}}
& \multicolumn{2}{c}{\textbf{Code}}
& \\
\cmidrule(lr){2-4}
\cmidrule(lr){5-7}
\cmidrule(lr){8-10}
\cmidrule(lr){11-13}
\cmidrule(lr){14-15}
\cmidrule(lr){16-17}
\textbf{Method}
& \textbf{NQA}
& \textbf{Qasper}
& \textbf{MF}
& \textbf{HPQA}
& \textbf{2WQA}
& \textbf{MQue}
& \textbf{GR}
& \textbf{QMS}
& \textbf{MN}
& \textbf{TREC}
& \textbf{TQA}
& \textbf{SSum}
& \textbf{PC}
& \textbf{PR}
& \textbf{LCC}
& \textbf{RBP}
& \textbf{Avg.} \\
\midrule

Full KV
& 30.99 & 47.39 & 56.31 & 59.80 & 51.33 & 33.51
& 35.19 & 25.34 & 26.99
& 29.00 & 91.71 & 40.77
& 10.70 & 100.00
& 53.31 & 47.23
& 46.22 \\

SnapKV
& 29.00 & 30.97 & 38.87 & 53.67 & 41.53 & 29.11
& 29.23 & 22.69 & 23.25
& 34.50 & 91.16 & 41.64
& 9.05 & 95.00
& \textbf{54.76} & 47.70
& 42.01 \\

PyramidKV
& 27.66 & 30.01 & 36.23 & 55.69 & 42.48 & 25.76
& 28.29 & 22.42 & 23.49
& 32.50 & 91.73 & \textbf{42.10}
& 9.05 & 97.50
& 54.73 & 47.40
& 41.69 \\

ExpectedAttention
& 29.77 & 38.94 & 44.37 & 55.93 & 43.02 & \textbf{34.13}
& 28.85 & 23.01 & 24.94
& 54.75 & 91.18 & 37.09
& 8.50 & 96.00
& 40.40 & \textbf{51.12}
& 43.88 \\

KeyDiff
& 29.75 & 32.87 & 42.45 & 53.05 & 39.11 & 25.40
& 29.49 & 23.66 & 22.78
& 51.50 & 91.61 & 40.85
& 11.06 & \textbf{99.50}
& 37.43 & 47.43
& 42.37 \\

CAPKV
& 29.04 & 39.18 & 43.66 & 52.76 & 45.45 & 24.96
& 29.28 & 23.67 & 25.21
& 30.25 & 91.33 & 37.67
& 10.00 & 67.50
& 44.04 & 50.91
& 40.31 \\

IndexMem
& 29.82 & 39.27 & 43.68 & 50.31 & 44.26 & 24.92
& 27.45 & 24.31 & 22.47
& 48.13 & 90.42 & 38.53
& 11.32 & 93.45
& 43.29 & 48.84
& 42.53 \\

MOMENTKV
& 28.00 & 26.97 & 37.42 & 51.77 & 39.90 & 26.89
& 27.86 & 22.14 & 24.01
& 31.50 & \textbf{92.23} & 40.88
& \textbf{12.05} & 95.50
& 52.64 & 46.05
& 40.99 \\

\textbf{CORE}
& \textbf{30.52}
& \textbf{44.29}
& \textbf{51.43}
& \textbf{58.17}
& \textbf{47.52}
& 30.17
& \textbf{32.59}
& \textbf{24.62}
& \textbf{25.62}
& \textbf{55.00}
& 91.46
& 40.62
& 10.00
& 98.00
& 49.33
& 49.50
& \textbf{46.18} \\

\bottomrule
\end{tabular}}
\end{table*}

\begin{table*}[t]
\centering
\vspace{-0.15in}
\caption{
Per-subtask RULER-16K results for Llama-3.1-8B at 75\% compression; Avg. is the unweighted mean over all 13 official subtasks.
}
\label{tab:ruler16k-per-subtask}
\scriptsize
\setlength{\tabcolsep}{2.6pt}
\resizebox{\textwidth}{!}{
\begin{tabular}{lcccccccccccccc}
\toprule
\textbf{Method}
& \textbf{S1}
& \textbf{S2}
& \textbf{S3}
& \textbf{MK1}
& \textbf{MK2}
& \textbf{MK3}
& \textbf{MV}
& \textbf{MQ}
& \textbf{VT}
& \textbf{CWE}
& \textbf{FWE}
& \textbf{QA1}
& \textbf{QA2}
& \textbf{Avg.} \\
\midrule

Full KV
& 100.00 & 100.00 & 100.00
& 99.60 & 100.00 & 98.80
& 99.10 & 99.10 & 99.72
& 89.32 & 90.27
& 81.00 & 57.00
& 93.38 \\

SnapKV
& 79.20 & 92.20 & 2.60
& 71.20 & 15.60 & 2.60
& 50.90 & 57.05 & 69.44
& 35.58 & 84.93
& 32.80 & 37.20
& 48.56 \\

PyramidKV
& 85.80 & 95.20 & 2.40
& 80.00 & 18.60 & 2.80
& 56.20 & 59.30 & 76.84
& 11.26 & 81.33
& 33.40 & 33.80
& 48.99 \\

ExpectedAttention
& \textbf{100.00} & 99.60 & 10.20
& 96.20 & 64.40 & 10.80
& 80.25 & 97.50 & 96.04
& 28.14 & 83.67
& 51.80 & 43.40
& 66.31 \\

KeyDiff
& \textbf{100.00} & \textbf{100.00} & \textbf{100.00}
& \textbf{99.60} & 26.20 & 0.00
& \textbf{99.10} & \textbf{99.45} & 98.64
& 71.52 & 87.80
& 44.80 & 44.60
& 74.75 \\

CAPKV
& \textbf{100.00} & \textbf{100.00} & 78.20
& 97.80 & 14.20 & 0.40
& 94.95 & 96.45 & 99.48
& 37.50 & 54.67
& 62.20 & 44.00
& 67.68 \\

IndexMem
& 99.20 & 87.20 & 33.80
& 70.80 & 97.60 & 3.40
& 78.40 & 70.10 & 95.00
& 88.64 & 83.13
& 75.60 & 55.60
& 72.19 \\

MOMENTKV
& 99.20 & 92.80 & 5.60
& 83.20 & 30.40 & 5.60
& 62.20 & 70.20 & 86.08
& 31.68 & 84.80
& 37.60 & 34.40
& 55.67 \\

CORE
& 95.40 & 91.60 & 34.00
& 86.60 & \textbf{98.20} & \textbf{12.40}
& 75.45 & 69.40 & \textbf{99.52}
& \textbf{89.16} & \textbf{89.07}
& \textbf{77.40} & \textbf{56.60}
& \textbf{74.98} \\

\bottomrule
\end{tabular}}
\vspace{-0.2in}
\end{table*}

\textbf{LongBench.}
Table~\ref{tab:longbench-per-task} reports all 16 tasks across
the six official categories.
CORE achieves the highest overall average among compressed methods,
outperforming the strongest compressed baseline, ExpectedAttention,
by 2.30 points, with gains spanning multiple QA and summarization tasks.

\textbf{RULER-16K.}
Table~\ref{tab:ruler16k-per-subtask} decomposes the same
75\% compression setting across all 13 RULER-16K subtasks:
single-needle retrieval (S1--S3), multi-key retrieval (MK1--MK3),
multi-value retrieval (MV), multi-query retrieval (MQ),
variable tracking (VT), common-word extraction (CWE),
frequent-word extraction (FWE), and question answering (QA1--QA2).
CORE achieves the highest overall average among compressed methods, narrowly ahead
of KeyDiff, while the per-task results show complementary
strengths across different retrieval configurations.

% ------------------------------------------------------------

% ------------------------------------------------------------

\subsection{Ultra-Long-Context Evaluation}
\label{app:ultra-long}

\begin{wraptable}{r}{0.43\textwidth}
\vspace{-0.6in}
\centering
\caption{Ultra-long-context accuracy (\%).}
\label{tab:ultra-long}
\footnotesize
\setlength{\tabcolsep}{2pt}
\renewcommand{\arraystretch}{0.9}
\begin{tabular*}{\linewidth}
{@{\extracolsep{\fill}}lccc@{}}
\toprule
\multirow{2}{*}{\textbf{Method}}
& \multicolumn{3}{c}{\textbf{KV budget}} \\
\cmidrule(l){2-4}
& \textbf{2K} & \textbf{4K} & \textbf{8K} \\
\midrule
Full KV & \multicolumn{3}{c}{32.02} \\
\midrule
SnapKV   & 25.12 & 26.60 & 27.09 \\
EA       & 24.63 & 25.12 & 26.11 \\
KeyDiff  & 25.62 & 26.11 & 25.62 \\
CAPKV    & 25.62 & 26.60 & 27.09 \\
IndexMem & 26.60 & 27.09 & 27.59 \\
\addlinespace[2pt]
\textbf{CORE}
& \textbf{31.03}
& \textbf{32.02}
& \textbf{32.02} \\
\bottomrule
\end{tabular*}
\vspace{-0.3in}
\end{wraptable}

On Llama-3.1-8B, we select 203 of the 503
LongBench-v2 \citep{bai-etal-2025-longbench} examples with complete prompt lengths of
32,768--131,056 tokens, using 4K-token chunked prefill,
greedy decoding, and a 16-token generation limit.
All methods share the same examples and settings;
KV budgets include question tokens and generation reserves.
CORE exceeds IndexMem by 4.43, 4.93, and 4.43 percentage
points across the three budgets.

\subsection{Hyperparameter Sensitivity}
\label{app:hyperparameter-sensitivity}

We examine four hyperparameters: the coverage weight $\beta$,
teacher entropy temperature $\tau_{\mathrm{ent}}$, boundary-loss weight
$\lambda_{\mathrm{bd}}$, and student allocation temperature $\tau_M$.
Experiments use Llama-3.1-8B on RULER-16K at 75\% compression,
with defaults $\beta_0=128$, $\tau_{\mathrm{ent},0}=8$,
$\lambda_{\mathrm{bd},0}=4$, and $\tau_{M,0}=1.0$.
For training-time sensitivity, we retrain CORE with
$\beta\in\{0,64,128,256,512\}$,
$\tau_{\mathrm{ent}}\in\{4,8,16\}$, and
$\lambda_{\mathrm{bd}}\in\{0,2,4,8\}$,
recomputing teacher targets when applicable while holding
training data and optimization budget fixed.
For inference-time sensitivity, we freeze the checkpoint and vary
$\tau_M\in\{0.5,1.0,2.0\}$.
For fixed logits and positive $\tau_M$, temperature scaling
preserves the Top-$B$ ordering and changes only the allocation
probabilities used by the compensation channel.

\begin{figure}[t]
	\centering
\includegraphics[width=\linewidth]{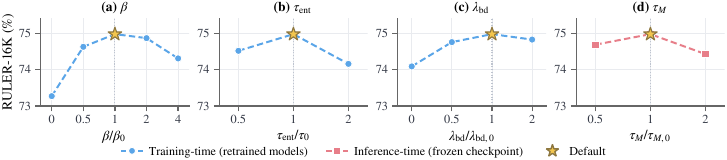}    \vspace{-0.3in}
\caption{
One-at-a-time hyperparameter sensitivity at 75\% compression.
}
	\label{fig:sensitivity}
    \vspace{-0.2in}
\end{figure}

Fig.~\ref{fig:sensitivity} shows that performance is more sensitive
to $\beta$ and $\lambda_{\mathrm{bd}}$: removing coverage or boundary
supervision causes clear degradation, while moderate nonzero values
around the defaults remain competitive. The $\beta$ sweep also suggests
that overly strong coverage weighting can hurt performance, reflecting
a balance between utility and coverage. In contrast, variations in
$\tau_{\mathrm{ent}}$ and $\tau_M$ lead to smaller changes over the
tested ranges. Overall, the default configuration lies within the
high-performing region of the evaluated settings.

\subsection{Shared and Alternative Allocation Strategies}
\label{app:shared-alternative-allocation}

\textbf{Shared versus independently learned write allocations.}
We compare CORE with $\mathrm{CORE}_{\mathrm{ind}}$, which retains
the same retention indexer but initializes a separate
$14\!\rightarrow\!128\!\rightarrow\!128\!\rightarrow\!1$ write scorer
at the start of Stage~II.
Its distribution $r_\psi$ replaces $\pi_\theta^{\mathcal C}$ in
Eq.~\ref{eq:evicted-mass} for write mass and conditional weights,
while Top-$B$ retention continues to use $\pi_\theta^{\mathcal C}$.
The write scorer is trained through $\mathcal L_{\mathrm{mem}}$
under the same Stage~II schedule, with matched data, initialization
of common components, memory capacity, and readout.
CORE outperforms its independent variant across all evaluated
settings (Table~\ref{tab:shared-independent-allocation}).

\begin{wraptable}{r}{0.42\textwidth}
\vspace{-0.3in}
\centering
\scriptsize
\setlength{\tabcolsep}{3pt}
\renewcommand{\arraystretch}{0.80}
\caption{CORE allocation replacement.}
\label{tab:allocation-ablation}
\resizebox{\linewidth}{!}{%
\begin{tabular}{@{}lcccc@{}}
\toprule
\multirow{2}{*}{\textbf{Variant}}
& \multicolumn{2}{c}{\textbf{RULER-4K}}
& \multicolumn{2}{c}{\textbf{RULER-16K}} \\
\cmidrule(lr){2-3}
\cmidrule(l){4-5}
& 75\% & 90\% & 75\% & 90\% \\
\midrule
\textbf{CORE}
& \textbf{80.60} & \textbf{62.44}
& \textbf{74.98} & \textbf{68.48} \\
$\text{CORE}_{\text{EA}}$
& 78.22 & 50.39 & 70.65 & 56.87 \\
$\text{CORE}_{\text{KeyDiff}}$
& 76.78 & 61.53 & 74.87 & 67.16 \\
\bottomrule
\end{tabular}%
}
\vspace{-0.2in}
\end{wraptable}

\textbf{Compatibility with alternative scoring rules.}
We replace $\pi_\theta^{\mathcal C}$ with distributions derived from
EA or KeyDiff scores for both Top-$B$ retention and evicted-mass
redistribution.
For each KV head, we normalize
$a_{h,i}^{1/T_{\mathrm{EA}}}$ or
$\exp(s_{h,i}/T_{\mathrm{KeyDiff}})$ over $\mathcal C$ and
average across KV heads.
Sink protection is applied separately, all-zero EA scores use a
uniform distribution, and temperatures are selected following
Appendix~\ref{app:validation-selection}.
Memory maps and gates are retrained under matched data,
initialization, optimization, and KV budgets.
CORE leads in all four settings
(Table~\ref{tab:allocation-ablation}).
The gap over $\mathrm{CORE}_{\mathrm{EA}}$ becomes particularly
pronounced at 90\% compression, while
$\mathrm{CORE}_{\mathrm{KeyDiff}}$ remains closer in several settings.
Together with the shared-versus-independent comparison, these results
help distinguish the roles of the shared retention--writing interface
and CORE's learned allocation.

\begin{table}[!htbp]
\centering
\vspace{-0.2in}
\caption{Training-source sensitivity on
Llama-3.1-8B.
Code Avg.\ averages LCC and RepoBench-P.
Higher scores are better; column-wise best results are bold.}
\label{tab:training_source_sensitivity}
\small
\setlength{\tabcolsep}{10pt}
\renewcommand{\arraystretch}{0.77}
\begin{tabular*}{\linewidth}
{@{\extracolsep{\fill}}lcccccc@{}}
\toprule
\multirow{2}{*}{\textbf{Training Source}}
& \multicolumn{2}{c}{\textbf{RULER-4K}}
& \multicolumn{2}{c}{\textbf{RULER-16K}}
& \textbf{LongBench}
& \textbf{Code Avg.} \\
\cmidrule(lr){2-3}
\cmidrule(lr){4-5}
\cmidrule(lr){6-6}
\cmidrule(l){7-7}
& 75\% & 90\% & 75\% & 90\% & 75\% & 75\% \\
\midrule
LongAlpaca-12k
& \textbf{80.60} & 62.44
& \textbf{74.98} & \textbf{68.48}
& \textbf{46.18} & 49.42 \\
OpenMathReasoning
& 80.18 & \textbf{62.69}
& 74.63 & 68.16
& 45.61 & 47.38 \\
SWE-smith-trajectories
& 79.72 & 61.83
& 74.41 & 67.54
& 45.83 & \textbf{50.12} \\
\bottomrule
\end{tabular*}
\vspace{-0.2in}
\end{table}

\subsection{Sensitivity to Training Data and Randomness}
\label{app:training-sensitivity}
\label{app:training-seed}

\textbf{Training-source sensitivity and downstream transfer.}
Using Llama-3.1-8B, we train CORE separately on LongAlpaca-12k, OpenMathReasoning~\citep{moshkov2025aimo2}, and SWE-smith-trajectories~\citep{yang2025swesmith}. OpenMathReasoning's CoT split provides mathematical problem--solution sequences. SWE-smith's tool split provides software-engineering trajectories with chronologically serialized messages, tool calls, and observations. Following Appendix~\ref{app:d.2}, we freeze the backbone and train the indexer and compensation modules with matched initializations and input-length distributions, regenerating coverage-teacher targets per source. Final checkpoints are evaluated following Appendix~\ref{app:d.3} on RULER-4K/16K at 75\%/90\% compression and all 16 LongBench tasks at 75\% compression, reporting the LCC/RepoBench-P average for code completion.
Table~\ref{tab:training_source_sensitivity} shows that
CORE maintains comparable aggregate performance across
general, mathematical, and software-engineering training sources.
Across the four RULER settings and the 16-task LongBench average,
deviations from LongAlpaca-12k are at most 0.94 points.
LongAlpaca-12k achieves the highest LongBench average and leads
three of four RULER settings, while SWE-smith improves the
code-subset average (50.12 vs.\ 49.42), indicating potential
for targeted domain adaptation.

\textbf{Training-seed sensitivity.}
To assess the trainable-baseline comparison's robustness
to training randomness, we independently retrained
Locret, TRIM-KV, KVP, IndexMem, and CORE on LongAlpaca-12k
with three seeds $\{42,43,44\}$, using
Mistral-7B-v0.2 and Llama-3.1-8B.
Each run used the same frozen pretrained backbone
and reinitialized method-specific trainable modules.
CORE and IndexMem repeated both training stages
independently per seed.
Across seeds, the training examples, held-out development set,
and each method's architecture, hyperparameters, and optimization
budget remained fixed; each run was evaluated at its final checkpoint.
Seeds varied initialization and training stochasticity,
including KVP policy sampling.
CORE's random projections, prototype bank, and Hadamard
rotation remained fixed.
Evaluation followed Appendix~\ref{app:d.3}, using identical
RULER-4K/16K instances, 75\%/90\% compression, and greedy decoding.
Table~\ref{tab:training-seed} illustrates three-run means and
sample standard deviations, with CORE leading eight
settings by 2.54--11.54 points over the strongest baseline.
CORE's standard deviations range from 0.04 to 0.27 points.

\begin{table}[!htbp]
\centering
\vspace{-0.2in}
\caption{Training-seed sensitivity on RULER.}
\label{tab:training-seed}
\small
\setlength{\tabcolsep}{4pt}
\renewcommand{\arraystretch}{0.85}

\begin{tabular*}{\linewidth}
{@{\extracolsep{\fill}}llcccc@{}}
\toprule
\multirow{2}{*}{\textbf{Model}}
& \multirow{2}{*}{\textbf{Method}}
& \multicolumn{2}{c}{\textbf{RULER-4K}}
& \multicolumn{2}{c}{\textbf{RULER-16K}} \\
\cmidrule(lr){3-4}
\cmidrule(l){5-6}
& & 75\% & 90\% & 75\% & 90\% \\
\midrule

\multirow{5}{*}{\shortstack[l]{Mistral-7B}}
& Locret
& $63.72_{\pm{0.26}}$ & $37.03_{\pm{0.09}}$
& $49.85_{\pm{0.35}}$ & $35.53_{\pm{0.03}}$ \\
& TRIM-KV
& $68.24_{\pm{0.35}}$ & $42.89_{\pm{0.07}}$
& $65.95_{\pm{0.05}}$ & $53.39_{\pm{0.07}}$ \\
& KVP
& $66.54_{\pm{0.32}}$ & $42.80_{\pm{0.29}}$
& $61.80_{\pm{0.06}}$ & $48.68_{\pm{0.35}}$ \\
& IndexMem
& $68.95_{\pm{0.16}}$ & $43.19_{\pm{0.12}}$
& $66.41_{\pm{0.09}}$ & $53.89_{\pm{0.16}}$ \\
& \textbf{CORE}
& $\mathbf{73.52}_{\pm{0.27}}$
& $\mathbf{54.73}_{\pm{0.13}}$
& $\mathbf{69.02}_{\pm{0.21}}$
& $\mathbf{56.43}_{\pm{0.19}}$ \\
\midrule

\multirow{5}{*}{\shortstack[l]{Llama-3.1-8B}}
& Locret
& $70.55_{\pm{0.22}}$ & $42.82_{\pm{0.09}}$
& $65.72_{\pm{0.28}}$ & $39.42_{\pm{0.26}}$ \\
& TRIM-KV
& $76.82_{\pm{0.07}}$ & $53.66_{\pm{0.10}}$
& $71.73_{\pm{0.10}}$ & $57.67_{\pm{0.29}}$ \\
& KVP
& $75.57_{\pm{0.22}}$ & $52.16_{\pm{0.34}}$
& $71.85_{\pm{0.16}}$ & $55.11_{\pm{0.18}}$ \\
& IndexMem
& $77.22_{\pm{0.07}}$ & $54.02_{\pm{0.23}}$
& $72.31_{\pm{0.20}}$ & $57.98_{\pm{0.04}}$ \\
& \textbf{CORE}
& $\mathbf{80.48}_{\pm{0.14}}$
& $\mathbf{62.51}_{\pm{0.12}}$
& $\mathbf{74.96}_{\pm{0.04}}$
& $\mathbf{68.57}_{\pm{0.14}}$ \\
\bottomrule
\end{tabular*}
\vspace{-0.1in}
\end{table}

\subsection{Full Efficiency Scaling and Parameter Cost}
\label{app:efficiency-details}

\begin{table}[htbp]
\centering

% ==================== Left: Latency ====================
\begin{minipage}[t]{0.49\linewidth}
\vspace{0pt}
\vspace{-0.2in}
\centering
\scriptsize
\setlength{\tabcolsep}{4pt}
\renewcommand{\arraystretch}{0.8}
\caption{Decode-only latency (ms/token).}
\label{tab:runtime-scaling}

\begin{tabular*}{\linewidth}
{@{\extracolsep{\fill}}llrrrr@{}}
\toprule
\textbf{Compression}
& \textbf{Method}
& \textbf{8K}
& \textbf{16K}
& \textbf{32K}
& \textbf{64K} \\
\midrule

\multirow{3}{*}{50\%}
& CAPKV    & 36.35 & 37.28 & 39.29 & 44.10 \\
& IndexMem & 39.92 & 40.93 & 43.39 & 49.15 \\
& CORE     & 34.69 & 35.59 & 37.27 & 41.41 \\
\midrule

\multirow{3}{*}{75\%}
& CAPKV    & 35.82 & 36.27 & 37.24 & 39.26 \\
& IndexMem & 39.33 & 39.91 & 40.95 & 43.32 \\
& CORE     & 34.36 & 34.73 & 35.55 & 37.29 \\
\midrule

\multirow{3}{*}{90\%}
& CAPKV    & 35.61 & 35.84 & 36.08 & 36.82 \\
& IndexMem & 39.02 & 39.28 & 39.62 & 40.49 \\
& CORE     & 34.10 & 34.21 & 34.55 & 35.22 \\
\bottomrule
\end{tabular*}
\end{minipage}\hfill
% ==================== Right: Memory ====================
\begin{minipage}[t]{0.49\linewidth}
\vspace{0pt}
\centering
\scriptsize
\vspace{-0.2in}
\setlength{\tabcolsep}{4pt}
\renewcommand{\arraystretch}{0.8}
\caption{Peak dynamic memory (GiB).}
\label{tab:memory-scaling}

\begin{tabular*}{\linewidth}
{@{\extracolsep{\fill}}llrrrr@{}}
\toprule
\textbf{Compression}
& \textbf{Method}
& \textbf{8K}
& \textbf{16K}
& \textbf{32K}
& \textbf{64K} \\
\midrule

\multirow{3}{*}{50\%}
& CAPKV    & 2.21 & 3.52 & 6.14 & 11.39 \\
& IndexMem & 2.62 & 3.96 & 6.65 & 12.02 \\
& CORE     & 2.58 & 3.91 & 6.57 & 11.88 \\
\midrule

\multirow{3}{*}{75\%}
& CAPKV    & 1.55 & 2.21 & 3.52 & 6.14 \\
& IndexMem & 1.95 & 2.62 & 3.96 & 6.65 \\
& CORE     & 1.92 & 2.58 & 3.91 & 6.57 \\
\midrule

\multirow{3}{*}{90\%}
& CAPKV    & 1.18 & 1.43 & 1.96 & 3.02 \\
& IndexMem & 1.57 & 1.82 & 2.37 & 3.46 \\
& CORE     & 1.54 & 1.79 & 2.33 & 3.41 \\
\bottomrule
\end{tabular*}
\end{minipage}

\end{table}

We evaluate Qwen3-14B using bfloat16 KV states and matched
retained-KV budgets. For prefill cache length $N$ and compression ratio $r$, the fixed budget
$B=\lfloor(1-r)N\rfloor$ includes protected tokens, and
$1\mathrm{K}=1024$.
Table~\ref{tab:runtime-scaling} reports mean per-token decoding
latency across 128 generated tokens, including boundary recompression
while excluding prefill and initial cache compression.
CORE is fastest throughout, especially at longer contexts.
At 64K, CORE requires 41.41, 37.29, and 35.22 ms/token at
50\%, 75\%, and 90\% compression, versus 44.10, 39.26, and 36.82
for CAPKV and 49.15, 43.32, and 40.49 for IndexMem.

Table~\ref{tab:memory-scaling} reports peak dynamic memory for
parameters, retained KV buffers, auxiliary states, and runtime
workspace, excluding frozen-backbone weights.
Measurements cover 1,024 post-compression decoding tokens,
with recompression every 128 tokens and prefill excluded.
Bfloat16 KV costs 160 KiB per position~\citep{yang2025qwen3},
while compaction temporarily retains both KV buffers.
Memory scales roughly linearly with retained context, falling
under stronger compression.
CAPKV has the smallest footprint, adding no learned parameters
or substantial auxiliary state.
CORE uses slightly less memory than IndexMem:
6.57 versus 6.65 GiB at 64K and 75\% compression.

\begin{wraptable}{r}{0.5\textwidth}
\centering
\vspace{-0.3in}
\scriptsize
\setlength{\tabcolsep}{4pt}
\renewcommand{\arraystretch}{0.8}
\caption{Additional parameter and runtime costs.}
\label{tab:parameter-cost}
\resizebox{\linewidth}{!}{%
\begin{tabular}{lccc}
\toprule
\textbf{Method}
& \textbf{Add. Params.}
& \textbf{Aux. State}
& \textbf{Online Overhead} \\
\midrule
CAPKV
& 0.00M
& $\approx 0$ GiB
& $+12.8\%$ \\
IndexMem
& $\approx 40$M
& 0.27 GiB
& $+24.1\%$ \\
CORE
& $\approx 40$M
& 0.22 GiB
& $+7.7\%$ \\
\bottomrule
\end{tabular}%
}
\end{wraptable}

Table~\ref{tab:parameter-cost} summarizes the parameter and online costs at 32K context and 75\% compression.  Auxiliary state excludes parameter storage, and overhead is measured against a 33.00 ms/token backbone-plus-compressed-attention reference. CORE's 0.22 GiB auxiliary state is smaller than IndexMem's 0.27 GiB, and its 7.7\% online overhead is the lowest among the three methods. CORE combines low latency with memory usage close to IndexMem while supporting coverage-aware retention and evicted-information compensation.

\begin{wraptable}{r}{0.60\textwidth}
\vspace{-0.3in}
\centering
\caption{Prefill latency and GPU-resident TTFT (s)
for Qwen3-14B at 75\% compression.}
\label{tab:prefill-ttft}
\scriptsize
\setlength{\tabcolsep}{2pt}
\renewcommand{\arraystretch}{0.95}
\begin{tabular*}{\linewidth}
{@{\extracolsep{\fill}}llrrrr@{}}
\toprule
Metric & Method & 8K & 16K & 32K & 64K \\
\midrule
\multirow{3}{*}{\shortstack[l]{Prefill\\latency}}
 & CAPKV    & 1.257 & 2.914 & 7.233 & 21.127 \\
 & IndexMem & 1.274 & 3.139 & 8.062 & 23.031 \\
 & CORE     & 1.388 & 3.597 & 9.402 & 27.562 \\
\midrule
\multirow{3}{*}{TTFT}
 & CAPKV    & 1.260 & 2.917 & 7.236 & 21.130 \\
 & IndexMem & 1.277 & 3.144 & 8.068 & 23.036 \\
 & CORE     & 1.393 & 3.602 & 9.406 & 27.567 \\
\bottomrule
\end{tabular*}
\vspace{-8pt}
\end{wraptable}
We compare CAPKV, IndexMem, and CORE on Qwen3-14B using
a single NVIDIA A100 80GB GPU, bfloat16 precision,
batch size 1, 75\% compression, identical 8K--64K-token
inputs, and the same backbone attention backend.
Table~\ref{tab:prefill-ttft} reports prefill latency,
covering prompt processing, scoring, and cache preparation,
including CORE's relation and coverage feature extraction,
indexer evaluation, Top-$B$ selection, cache compaction,
and initial memory writing.
CORE's additional overhead is concentrated in initial
prompt processing and cache preparation; final-position
vocabulary projection and greedy token selection add
only 4--5 ms once the cache is ready.
This upfront computation supports more effective cache
retention and memory utilization, representing a reasonable
trade-off when considered alongside task-performance gains.

\section{Limitations}
\label{app:statements}

Although CORE performs well across three backbones and multiple
KV budgets, several questions remain. The stability analysis
focuses on a fixed query; extending it to full autoregressive
trajectories remains future work. The largest tested backbone
has 14B parameters, so performance at larger scales remains
untested. Efficiency measurements cover specific context lengths,
hardware, and recompression schedules, and the resulting
trade-offs may vary in other configurations. Future work will
extend the analysis and evaluate larger models under broader
inference settings.

\end{document}

%% file: math_commands.tex
\usepackage{amsmath,amsfonts,bm}

\def\eqref#1{equation~\ref{#1}}
\def\1{\bm{1}}

\DeclareMathAlphabet{\mathsfit}{\encodingdefault}{\sfdefault}{m}{sl}
\SetMathAlphabet{\mathsfit}{bold}{\encodingdefault}{\sfdefault}{bx}{n}